\documentclass[twoside,11pt]{article}
\usepackage{style, styles, rawfonts}

\newif\iflong\longtrue

\usepackage[utf8]{inputenc}
\usepackage{amssymb}
\usepackage{amsmath}
\usepackage{enumerate}
\usepackage{dsfont}
\usepackage{booktabs}
\usepackage{color}
\usepackage{tikz}
\usepackage[disable]{todonotes}
\usepackage[standard, amsmath]{ntheorem}
\usepackage{placeins}
\usepackage[]{algorithm2e}

\renewtheorem{definition}{Definition}
\renewtheorem{proposition}[definition]{Proposition}
\renewtheorem{corollary}[definition]{Corollary}
\renewtheorem{theorem}[definition]{Theorem}
\renewtheorem{lemma}[definition]{Lemma}
\newtheorem{claim}{Claim}

\newcommand{\Oh}{\mathcal{O}}
\newcommand{\Mo}{\mathcal{M}} 
\newcommand{\Sk}{\mathcal{S}}
\newcommand{\Fa}{\mathcal{F}} 
\newcommand{\poly}{\text{\normalfont{poly}}}
\newcommand{\badstuffhappens}{\ensuremath{\text{\rm NP} \subseteq \text{\rm coNP} / \text{\rm poly}}}

\newcommand{\todok}[1]{\todo[backgroundcolor=red,linecolor=black]{ #1}}

\DeclareMathOperator*{\argmax}{\arg\!\max}

\jairheading{x}{x}{x-x}{x/x}{x/x}
\title{Learning Bayesian Networks Under Sparsity Constraints: A~Parameterized Complexity Analysis}

\author{\name Niels Grüttemeier \email niegru@informatik.uni-marburg.de \\
       \name Christian Komusiewicz \email komusiewicz@informatik.uni-marburg.de \\
       \addr Hans-Meerwein-Straße 6,\\
		35032 Marburg, Germany
       }
\begin{document}


\maketitle

\begin{abstract}
  
  We study the problem of learning the structure of an optimal Bayesian network when additional constraints are posed on the network or on its moralized graph. More precisely, we consider the constraint that the network or its moralized graph are close, in terms of vertex or edge deletions, to a sparse graph class~$\Pi$. For example, we
  show that learning an optimal network whose moralized graph has vertex deletion distance at most~$k$ from a graph with maximum degree 1 can be
  computed in polynomial time when~$k$ is constant. This extends previous work that gave an algorithm with such a 
  running time for the vertex deletion distance to edgeless graphs~[Korhonen \&
  Parviainen, NIPS 2015]. We then show that further extensions or improvements are presumably
  impossible. For example, we show that learning optimal networks where the network or its moralized graph have maximum degree~$2$
  or connected components of size at most~$c$,~$c\ge 3$, is NP-hard. Finally, we show that learning an optimal
  network with at most~$k$ edges in the moralized graph presumably has no~$f(k)\cdot |I|^{\Oh(1)}$-time algorithm and that, in contrast, an optimal network with at
  most~$k$ arcs  can be computed in~$2^{\Oh(k)}\cdot
  |I|^{\Oh(1)}$~time where~$|I|$ is the total input~size.
\end{abstract}


\subsection*{Acknowledgment}

We would like to thank our colleague Nils Morawietz (Philipps-Universität Marburg) for his helpful discussions that led to the proof of Theorem~\ref{Theorem: Bounded-VC W[2]-h}. A preliminary version of this work appeared in \emph{Proceedings of the Twenty-Ninth International Joint Conference on Artificial Intelligence, (IJCAI '20)}, pages 4245--4251. The full version contains all missing proofs, new results for~\textsc{BNSL} with constraints on the skeleton, and an improved algorithm for~\textsc{$(\Pi_1+v)$-Moral BNSL} leading to a slightly better running time.

\section{Introduction}
\todok{degree-2-modulater is the set, not the number}
Bayesian networks are graphical models for probability distributions in which the presence of statistical
dependencies between a set of random variables are represented via a directed acyclic graph
(DAG)~$D=(N,A)$ over a set~$N$ of~$n$ random variables~\cite{D09}. An arc \iflong from a vertex~$u$ to a vertex~$v$ \fi in a Bayesian
network means that the distribution of~$v$ depends on the value of~$u$. Once we have
obtained a Bayesian network, one may infer the distribution of some random variables given the
values of other random variables.

First, however, one needs to learn the network from observed data. An important step herein
is to learn the \emph{structure} of the network, that is, the arc set of the corresponding DAG. This problem is known as \textsc{Bayesian Network Structure Learning (BNSL)}. In BNSL, one is given for each network vertex~$v$ and each set of possible
parents of~$v$ a parent score and the goal is to learn an acyclic network with a maximal sum of parent scores. To  represent the observed data as
closely as possible, it may seem appropriate to learn \iflong a tournament, that is, \fi a DAG in which
every pair of vertices~$u$ and~$v$ is connected either by the arc~$(u,v)$ or by the
arc~$(v,u)$. There are, however, several reasons why \iflong learning a tournament-like DAG \fi 
should be avoided~\iflong(For a detailed discussion we refer to the book by~\citeA{D09})\fi: First, such a network
gives no information about which variables are conditionally independent. Second, including too
many dependencies \iflong in the model \fi makes the model vulnerable to overfitting. Finally, the problem
of inferring distributions on a given Bayesian network is intractable when the DAG is \iflong tournament-like\fi.
More precisely, \iflong the inference problem \fi on Bayesian networks is NP-hard~\cite{C90}. 
When the  network is tree-like, however, efficient inference algorithms are possible: If the moralized graph has small treewidth, the inference task can be solved more efficiently~\cite{D09}; the moralized graph of a network~$D$ is the undirected graph on the same
vertex set that is obtained by adding an edge between each pair of vertices that is adjacent or has a common child
in~$D$.

Motivated by \iflong these reasons for avoiding tournament-like networks and instead aiming for
tree-like networks,\fi{} it has been proposed to learn optimal networks under structural constraints
that guarantee that the network or its moralized graph is
tree-like~\cite{EG08,KP13,KP15,CL68,D99,GKLOS15}.  We continue this line of research,
focusing on exact algorithms with worst-case running time guarantees. In other words, we want
to find out for which structural constraints there are fast algorithms for learning optimal
Bayesian networks under these constraints and for which constraints this is
presumably impossible. 

\paragraph{Known Results.}
The problem of learning a Bayesian network without structural constraints, which we call 
\textsc{Vanilla-BNSL}, is NP-hard~\cite{C95} and can be
solved in~$2^n n^{\Oh(1)}$ time by dynamic programming over all subsets
of~$N$~\cite{OM03,SM06}.

When the network is restricted to be a \iflong branching, that is, a \fi directed tree in which every
vertex has indegree at most one, then an optimal network can be computed in polynomial
time~\cite{CL68,GKLOS15}. \iflong Note that learning a more restricted Bayesian network is not necessarily easier: While learning a branching is solvable in polynomial time, the problem becomes NP-hard if we aim to learn a directed path~\cite{M01}.\fi

On the negative side, BNSL where the moralized graph of the network is restricted to have
treewidth at most~$\omega$ is NP-hard for every fixed~$\omega \geq 2$ and can be solved
in~$3^n n^{\omega+\Oh(1)}$ time~\cite{KP13}.  Finally, Korhonen and Parviainen~\cite{KP15}
considered a structural constraint that restricts the treewidth of the moralized graph by
restricting the size of its vertex cover. A \emph{vertex cover} in a graph~$G$ is a vertex
set~$S$ such that every edge of~$G$ has at least one endpoint in~$S$. Korhonen and Parviainen~\cite{KP15} showed that \textsc{BNSL} where
the moralized graph is restricted to have a vertex cover of size at most~$k$ can be solved
in~$4^k\cdot n^{2k+\Oh(1)}$ time~\cite{KP15}. Since having a bounded vertex cover---we refer to graphs with this property as bounded-vc graphs---implies
that the graph has bounded treewidth, the networks that are learned by \textsc{BNSL} with
bounded-vc moralized graphs allow for fast inference algorithms.  An
algorithm with running time~$f(k)\cdot |I|^{\Oh(1)}$ is unlikely for this \textsc{BNSL}
variant, since it is W[1]-hard with respect to the parameter~$k$~\cite{KP15}. Here,~$|I|$
denotes the total input size. In other
words, 
it seems necessary that the degree of the running time polynomial depends on~$k$.


\paragraph{Our Results.}
\begin{table}[t]
\caption{An overview of the parameterized complexity of constrained BNSL problems for an input parameter~$k$ that upper-bounds structural parameters in the resulting skeleton or moralized graph. The \emph{distance to degree 2} is the minimum size of a vertex set~$S$, such that after removing~$S$, the maximum degree is 2. The~\emph{$c$-component order connectivity ($c$-COC)} is the minimum size of a vertex set~$S$, such that after deleting~$S$, every connected componenet has order at most~$c$.}
\begin{center}
\begin{tabular}{lll}
\toprule
\textbf{Bounded by~$\mathbf{k}$}&\textbf{Skeleton}&\textbf{Moralized Graph}\\
\midrule
Treewidth &NP-h for~$k=1$ &NP-h for~$k=2$\\
&\cite{D99}&\cite{KP13}\\

\midrule

Vertex cover number & XP (Thm~\ref{Theorem: VC XP Algo}) & XP~\cite{KP15} \\
& W[2]-h (Thm~\ref{Theorem: Bounded-VC W[2]-h}) & W[2]-h (Cor~\ref{Cor: Bounded-VC W[2]-h})\\

\midrule

Dissociation number &W[2]-h (Thm~\ref{Theorem: Bounded-VC W[2]-h}) & XP (Thm~\ref{Theorem: XP Algo}) \\
&&W[2]-h (Cor~\ref{Cor: Bounded-VC W[2]-h})\\

\midrule

Distance to degree~2 &NP-h for~$k=0$ (Thm \ref{Theorem: Bounded-DelToDeg2 NPh}) &NP-h for~$k=0$ (Thm \ref{Theorem: Bounded-DelToDeg2 NPh})\\

\midrule

$c$-COC for~$c\geq 3$ &NP-h for~$k=0$ (Thm \ref{Theorem: Bounded-3coc NP-h}) &NP-h for~$k=0$ (Thm \ref{Theorem: Bounded-3coc NP-h})\\

\midrule

Number of edges &FPT (Cor~\ref{Cor: BA-BNSL FPT})  & XP (Prop~\ref{Prop: Bounded-Arc BNSL in XP}) \\
&no~$k^{\Oh(1)}$ kernel (Cor~\ref{Cor: BNSL No Poly Kernel for n})&W[1]-h (Thm~\ref{Theorem: Bounded-Edges BNSL W[1]-h})\\

\midrule

Feedback edge set & NP-h for~$k=0$ & W[1]-h (Thm~\ref{Theorem: B-FES BNSL W1h})\\
&\cite{D99}&\\
\bottomrule
\end{tabular}
\end{center}
\label{tab:results}
\end{table}

Extending previous work, we provide an almost complete picture of the
parameterized and classical complexity of BNSL with respect to several constraints that
guarantee tree-like networks or moralized graphs. Since the constraints are formulated in terms of undirected graphs, we will refer to the undirected underlying graph of a network as its \textit{skeleton}. An overview of our results and previous results for the considered problems is given in Table~\ref{tab:results}.

The results for \textsc{BNSL} with bounded-vc moralized graphs~\cite{KP15} form the starting point for our
work. We show that BNSL with bounded-vc skeletons can be solved in polynomial time when the vertex cover size bound~$k$ is a constant. Moreover, we show
that, as for bounded-vc moralized graphs, an algorithm with running time~$f(k)\cdot |I|^{\Oh(1)}$ is unlikely since BNSL with bounded-vc skeletons is W[2]-hard.

After complementing the results for BNSL  with bounded-vc moralized graphs by
results for its skeleton counter part, we consider further, related structural
constraints.  To this end, we take the following alternative view of vertex covers: A
graph has a vertex cover of size~$k$ if and only if it can be transformed into an edgeless
graph by~$k$ vertex deletions. Thus, in \textsc{BNSL} with bounded-vc moralized graphs we
learn a network whose moralized graph is close, in terms of the number of vertex
deletions, to a sparse graph class. We investigate whether there are further positive
examples for such constrained network learning problems.

First, we consider the constraint that the skeleton or the moralized graph can be
transformed into a graph with maximum degree~$1$ by at most~$k$ vertex deletions. This
property is also known as having \emph{dissociation number} at most~$k$ and we refer to
graphs with this property as bounded-diss-number graphs in the following. We show that
under this constraint, BNSL with bounded-diss-number moralized graphs can be solved
in~$n^{3k} \cdot k^{\Oh(k)} \cdot |I|^{\Oh(1)}$~time and thus in polynomial time for every constant
value of~$k$. This extends the result for \textsc{BNSL} with bounded-vc moralized graphs
in the following sense: the value of~$k$ can be arbitrarily smaller than the vertex cover
number of the resulting network and thus for fixed~$k$ our algorithm can learn an optimal
network for a larger class of graphs than the algorithm for \textsc{BNSL} with bounded-vc
moralized graphs. Observe that moralized graphs with bounded dissociation number still
have bounded treewidth and thus inference on the learned networks will still be solvable
efficiently. On the negative side, we show that an algorithm with running
time~$f(k)\cdot |I|^{\Oh(1)}$ is unlikely since the problem is W[2]-hard. The latter
hardness result also holds for BNSL with bounded-diss-number skeletons; we did
not obtain a positive result for this case, however.

We then consider two further constraints that are related to the dissociation number: We
show that learning an optimal network whose skeleton or moralized graph has maximum
degree~$2$ is NP-hard and that learning an optimal network in which every component of the
skeleton or the moralized graph has at most $c$~vertices,~$c\ge 3$, is
NP-hard. The latter constraint is related to the dissociation number since in a graph with
maximum degree 1 every connected component has at most two vertices.

Next, we consider constraints that are formulated in terms of edge sets of the skeleton or the moralized graph. We show that optimal networks with at most~$k$ arcs can be found in time~$2^{\Oh(k)}\cdot |I|^{\Oh(1)}$. In contrast, when we aim to compute an optimal network whose moralized graph has at most~$k$
edges, an~$f(k)\cdot |I|^{\Oh(1)}$-time algorithm is unlikely.  Thus, putting structural constraints on the
  moralized graph may make the problem much harder than putting similar structural
  constraints on the skeleton.

Furthermore, we consider the case where the edge deletion distance to trees is measured, in
other words, the case where the skeleton or the moralized graph have a feedback edge set
of size at most~$k$. BNSL with tree skeletons is known
as \textsc{Polytree Learning}. Thus, the learning problem for skeletons with feedback edge
sets of size at most~$k$ is NP-hard even for~$k=0$~\cite{D99}. For BNSL with moralized
graphs with feedback edge sets of size at most~$k$, we obtain the first hardness result:
we show that an algorithm with running time~$f(k) \cdot |I|^{\Oh(1)}$ is unlikely, since
the problem is W[1]-hard; an algorithm with running time~$n^{f(k)}$ is however still possible.

Finally, we obtain a further hardness result for \textsc{Vanilla-BNSL}: Under standard
assumptions in complexity theory, it is impossible that we can transform a given instance of
\textsc{Vanilla-BNSL} in polynomial time to an equivalent one of size~$n^{\Oh(1)}$. Thus, it is sometimes necessary to keep an exponential number of parent scores to
compute an optimal network.

Altogether, our results reveal that the difficulty of the learning problem may differ depending on whether we put the constraints on the skeleton or the moralized graph. Moreover, more general networks than those with bounded-vc moralized graphs can be computed efficiently. The room for generalization seems, however, very limited as even learning networks with constant degree or constant component size is NP-hard. 

\section{Preliminaries}
\iflong
\subsection{Notation}
\fi
A directed graph~$D=(N,A)$ consists of a \emph{vertex set}~$N$ and an \emph{arc set}~$A \subseteq N \times N$. Let~$D=(N,A)$ be a directed graph. If~$D$ does not contain directed cycles, then~$D$ is a \emph{directed acyclic graph}~(DAG). An arc~$(u,v) \in A$ is called \emph{incoming arc into~$v$} and~\emph{outgoing arc from~$u$}. \iflong Given a vertex~$v$, the number of incoming arcs into~$v$ is called~\emph{in-degree of~$v$}, and the number of outgoing arcs from~$v$ is the~\emph{out-degree of~$v$}. \fi A vertex without incoming arcs is a~\emph{source}. A vertex without outgoing arcs is a~\emph{sink}.
The set~$P^A_v:=\{u \in N \mid (u,v)\in A\}$ is called \emph{parent set of~$v$}. The vertices in~$P^A_v$ are called \emph{parents of~$v$} and for every~$u \in P^A_v$, the vertex~$v$ is called \emph{child of~$u$}. We call~$v_1$ an \emph{ancestor of~$v_\ell$} and~$v_\ell$ a \emph{descendant of~$v_1$} if there is a directed path~$(v_1, v_2, \dots, v_\ell)$ in~$D$.

An undirected graph~$G=(V,E)$ consists of a vertex set~$V$ and an \emph{edge set}~$E \subseteq \{ \{u,v\} \mid u,v \in V\}$.
For a vertex~$v \in V$, we write~$N_G(v):=\{u \mid \{u,v\} \in E\}$ to denote the neighborhood of~$v$ in~$G$. The \emph{degree of a vertex~$v$} is denoted by~$\deg_G(v):=|N_G(v)|$. \iflong For~$V_1, V_2 \subseteq V$, we write~$E_G(V_1,V_2):= \{\{v_1,v_2\} \in E \mid v_1 \in V_1, v_2 \in V_2\}$ for the set of edges between~$V_1$ and~$V_2$. Moreover we set~$E_G(K):=E_G(K,K)$. \fi Given an edge set~$E' \subseteq E$, we let~$G \setminus E'$ denote the graph we obtain after deleting the edges of~$E'$ from~$G$. Given a vertex set~$V' \subseteq V$, we let~$G-V'$ denote the graph we obtain after deleting the vertices in~$V'$ and their incident edges from~$G$. \iflong A set~$T \subseteq E$ is called \emph{feedback edge set} if~$G \setminus T$ contains no cycles. The size of a smallest possible feedback edge set for~$G$ is the~\emph{feedback edge number of~$G$}. \fi A set~$S \subseteq V$ is called \emph{dissociation set} if~$G-S$ has maximum degree 1. The size of a smallest possible dissociation set for~$G$ is the~\emph{dissociation number of~$G$}.

A \emph{graph class}~$\Pi$ is a set of undirected graphs. For a graph class~$\Pi$ and~$k \in \mathds{N}$, let~$\Pi + kv \iflong := \{G=(V,E) \mid \exists V' \subseteq V: (|V'| \leq k \land G-V' \in \Pi)\} \fi$ denote the class of graphs that can be transformed into a graph in~$\Pi$ by performing at most~$k$ vertex deletions. Analogously, we define~$\Pi+ ke \iflong := \{G=(V,E) \mid \exists E' \subseteq E: (|E'| \leq k \land G \setminus E' \in \Pi)\} \fi$ as the class of graphs that can be transformed into a graph in~$\Pi$ by performing at most~$k$ edge deletions. We call $\Pi$~\emph{monotone} if~$\Pi$ is closed under edge- and vertex deletions. Note that~$\Pi$ being monotone implies that for every~$k \in \mathds{N}_0$, the graph classes~$\Pi+kv$ and~$\Pi+ke$ are monotone. 

\subsection{Bayesian Network Structure Learning}

\paragraph{Problem Definitions.}
Given a vertex set~$N$, we call a family~$\Fa=\{f_v:2^{N \setminus \{v\}} \to \mathds{N}_0 \mid v \in N\}$ a family of \emph{local scores} for~$N$. Intuitively, for a vertex~$v \in N$ and some~$P \in 2^{N \setminus \{v\}}$, the value~$f_v(P) \in \mathds{N}_0$ represents the score we obtain if we choose exactly the vertices of~$P$ as parents for~$v$. Given a vertex set~$N$, local scores~$\Fa$, and some integer~$t \in \mathds{N}_0$, an arc set~$A \subseteq N \times N$ is called \emph{$(N,\Fa,t)$-valid} if~$(N,A)$ is a DAG and~$\sum_{v \in N} f_v(P^A_v) \geq t$. In~\textsc{Bayesian Network Structure Learning (Vanilla-BNSL)}, the input is a vertex set~$N$, local scores~$\Fa$, and an integer~$t$. The question is, whether there exists an~$(N, \Fa, t)$-valid arc set.

In this work, we study~\textsc{Vanilla-BNSL} under additional sparsity constraints. These sparsity constraints are posed on the skeleton and the moralized graph~\cite{EG08} of the network. Let~$D:=(N,A)$ be a DAG. The \emph{skeleton of~$D$} is the undirected graph~$\Sk(D):=(N,E)$, with~$E:=\{\{u,v\} \mid (u,v) \in A\}$. The \emph{moralized graph of~$D$} is the undirected graph~$\Mo(D):=(N, E_1 \cup E_2)$ where the edge set is defined by~$E_1 := \{ \{u,v\} \mid (u,v) \in A\}$ and~$E_2 := \{ \{u,v\} \mid u \text{ and }v \text{ have a common child in }D \}$. The edges in~$E_2$ are called~\emph{moral edges}. Given a DAG~$(N,A)$ we write~$\Sk(N,A):=\Sk((N,A))$ and~$\Mo(N,A):=\Mo((N,A))$ for sake of readability. The problems are defined as~follows.

\begin{center}
	\begin{minipage}[c]{.9\linewidth}
          \textsc{$(\Pi+v)$-Skeleton Bayesian Network Structure Learning \\($(\Pi+v)$-Skeleton BNSL)}\\
          \textbf{Input}: A set of vertices~$N$, local scores~$\Fa=\{f_v \mid v \in N\}$, and two integers~$t,k \in \mathds{N}_0$.\\
          \textbf{Question}: Is there an $(N,\Fa,t)$-valid arc set~$A \subseteq N \times N$ such that~$\Sk(N,A) \in \Pi + kv$?
	\end{minipage}
\end{center}

\begin{center}
	\begin{minipage}[c]{.9\linewidth}
          \textsc{$(\Pi+v)$-Moral Bayesian Network Structure Learning \\($(\Pi+v)$-Moral BNSL)}\\
          \textbf{Input}: A set of vertices~$N$, local scores~$\Fa=\{f_v \mid v \in N\}$, and two integers~$t,k \in \mathds{N}_0$.\\
          \textbf{Question}: Is there an $(N,\Fa,t)$-valid arc set~$A \subseteq N \times N$ such that~$\Mo(N,A) \in \Pi + kv$?
	\end{minipage}
\end{center}

Furthermore, we define the problems~\textsc{$(\Pi+e)$-Skeleton-BNSL} and \textsc{$(\Pi+e)$-Moral-BNSL} on the same input and we ask if there exists an~$(N,\Fa,t)$-valid arc set~$A$ such that~$\Sk(N,A) \in \Pi + ke$ or~$\Mo(N,A) \in \Pi + ke$, respectively. Given a graph class~$\Pi$, we refer to all problems described above as~\emph{constrained BNSL problems for~$\Pi$}. For a constrained BNSL problem we refer to the constraint on~$\Sk$ or~$\Mo$ as \emph{sparsity constraint}. Given an instance~$I$ of a constrained BNSL problem for some~$\Pi$, we call the requested arc set~$A$ a \emph{solution} of~$I$. Note that, if~$\Pi$ is monotone and contains infinitely many graphs and~$k=n$, then the sparsity constraints~$\Sk(N,A) \in \Pi + kv$ and~$\Mo(N,A) \in \Pi + kv$ always hold, since the empty graph belongs to~$\Pi$. Moreover, if~$\Pi$ is monotone and contains infinitely many graphs and~$k=n^2$, then the sparsity constraints~$\Sk(N,A) \in \Pi + ke$ and~$\Mo(N,A) \in \Pi + ke$ always hold, since every edgeless graph belongs to~$\Pi$. Hence,~all problems considered in this work are generalizations of~\textsc{Vanilla-BNSL} and thus NP-hard for every monotone and infinite~$\Pi$. For formal reasons, the problems are stated as decision problems. However, the algorithms presented in this work solve the corresponding optimization problem within the same running time.

 \paragraph{Input Representation.} 
Throughout this work, we let~$n:=|N|$ denote the number of vertices given in an instance~$I=(N,\Fa,t,k)$ of a constrained BNSL problem. Furthermore, we assume that for~$N=\{v_1, \dots, v_n\}$, the local scores~$\Fa$ are given as a two-dimensional array~$\Fa:=[Q_{1}, Q_{2}, \dots, Q_{n}]$, where each~$Q_i$ is an array containing all triples~$(f_{v_i}(P), |P|, P)$ where~$f_{v_i}(P)>0$ or~$P=\emptyset$. Note that the triple~$(f_{v_i}(P), |P|, P)$ of a non-empty parent set~$P$ is only part of the input if its local score~$f_{v_i}(P)$ is not 0. This input representation is known as \emph{non-zero representation}~\cite{OS13}. The size~$|\Fa|$ is then defined as the number of bits we need to store this two-dimensional array. As the \emph{size of~$I$} we define~$|I|:= n + |\Fa| + \log(t) + \log(k)$. 

 \paragraph{Basic Observations.} 
Let~$I$ be a yes-instance of a constrained BNSL problem. We call a solution~$A$ for~$I$ \emph{nice} if~$f_v(P^A_v) \leq f_v(\emptyset)$ implies~$P^A_v = \emptyset$. In this work, we consider constrained BNSL problems for some monotone graph classes~$\Pi$. We next show that in these cases every yes-instance has a nice solution~$A$.
 
\begin{proposition} \label{Prop: Potential Parents}
Let~$\Pi$ be a monotone graph property, and let~$(N, \Fa, t, k)$ be a yes-instance of a constrained BNSL problem for~$\Pi$. Then, there exists a nice solution~$A$ for~$(N, \Fa, t, k)$. 
\end{proposition}

\iflong
\begin{proof}
Let~$A$ be a solution for~$I:=(N, \Fa, t, k)$ such that there exist vertices~$v_1, \dots, v_\ell \in N$ with~$P_{v_i}^A \neq \emptyset$ and~$f_{v_i}(P_{v_i}^A)\leq f_{v_i}(\emptyset)$. We then set~$A' := A \setminus \{(u, v_i) \mid u \in N, i \in \{1, \dots, \ell\}\}$. Observe that~$P_{v_i}^{A'} = \emptyset$ for all~$i \in \{1, \dots, \ell\}$. Moreover,~$f_{v}(P^{A'}_{v}) \geq f_{v}(P^{A}_{v})$ for every~$v \in N$ and~$(N,A')$ is a DAG. Therefore,~$A'$ is~$(N, \Fa, t)$-valid. Finally, since~$\Pi$ is monotone and~$\Sk(N,A)$ (or~$\Mo(N,A)$, respectively) satisfies the sparsity constraint,~$\Sk(N,A')$ (or~$\Mo(N,A')$, respectively) also satisfies the sparsity constraint. $\hfill \Box$
\end{proof}
\fi

Observe that an instance~$I:=(N,\Fa,t,k)$ of a constrained BNSL problem for some monotone~$\Pi$ is a trivial yes-instance if~$\sum_{v \in N} f_v(\emptyset) \geq t$, since the empty arc set is a solution of~$I$. Throughout this work, we assume that for a non-trivial instance of a constrained BNSL problem it holds that~$f_v(\emptyset)=0$ for every~$v \in N$. With the next proposition we assure that every instance can be preprocessed in~$\Oh(|\Fa|)$ time into an instance that satisfies our assumption.

\begin{proposition} \label{Prop: Translation}
Let~$\Pi$ be a graph class, and let~$I:=(N,\Fa,t,k)$ be an instance of a constrained BNSL problem for~$\Pi$ where~$t \geq \sum_{v \in N} f_v(\emptyset)$. Then, there exist~$\Fa':=\{f_v' \mid v \in N\}$ with~$f'_v(\emptyset)=0$ for every~$v\in N$ and~$t' \in \mathds{N}_0$, such that an arc set~$A$ is a nice solution for~$I$ if and only if~$A$ is a nice solution for~$I':= (N,\Fa',t',k)$. Furthermore,~$I'$ can be computed in~$\Oh(|\Fa|)$~time.
\end{proposition}

\begin{proof}

Let~$v \in N$. We define the new local scores~$f_v'$ by setting~$f_v'(P):=f_v(P)-f_v(\emptyset)$, if~$f_v(P) \geq f_v(\emptyset)$, and~$f_v'(P):=0$ otherwise. Note that~$f'_v(\emptyset)=0$ for all~$v \in N$. Furthermore, we set~$t' := t - \sum_{v \in N} f_v(\emptyset)$. Obviously,~$\Fa'$ and~$t'$ can be computed in~$\Oh(|\Fa|)$ time by iterating over the two-dimensional array representing~$\Fa$. Moreover,~$t' \geq 0$ since~$t \geq \sum_{v \in N} f_v(\emptyset)$. We next show that~$A \subseteq N \times N$ is a nice solution for~$I$ if and only if~$A$ is a nice solution for~$I'$.

$(\Rightarrow)$ Let~$A$ be a nice solution for~$I$. Obviously,~$(N,A)$ is a DAG and the sparsity constraint is satisfied. 
Furthermore, we have
\begin{align*}
\sum_{v \in N} f'_v(P^A_v) &\geq \sum_{v \in N} (f_v(P^A_v) - f_v(\emptyset))\\
& \geq t - \sum_{v \in N} f_v(\emptyset) = t'.
\end{align*}
It remains to show that~$A$ is nice for~$I'$. To this end, let~$f'_v(P^A_v) \leq f'_v(\emptyset)$. We conclude~$f'_v(P^A_v)=0$ and therefore~$f_v(P^A_v) \leq f_v(\emptyset)$. Since~$A$ is nice for~$I$, we conclude~$P^A_v=\emptyset$. Hence,~$A$ is a nice solution for~$I'$.

$(\Leftarrow)$ Conversely, let~$A$ be nice for~$I'$. We show that~$A$ is a nice solution for~$I$. Obviously~$(N,A)$ is a DAG and the sparsity constraint is satisfied. Hence, it remains to show that~$\sum_{v \in N} f_v(P^A_v) \geq t$ and that~$A$ is nice for~$I$. 

To this end, we first show that for every~$v \in N$ it holds that~$f'_v(P^A_v) = f_v(P^A_v) - f_v(\emptyset)$. Assume towards a contradiction that there exists some~$v \in N$ such that~$f'_v(P^A_v) \neq  f_v(P^A_v) - f_v(\emptyset)$. It then follows by the definition of~$\Fa'$, that~$f_v(P^A_v) < f_v(\emptyset)$ and that~$f_v'(P^A_v)=0$. Note that~$f_v'(P^A_v)=0$ implies~$f_v'(P^A_v) \leq f_v'(\emptyset)$ and therefore~$P^A_v=\emptyset$ since~$A$ is nice for~$I'$. This contradicts the fact that~$f_v(P^A_v) < f_v(\emptyset)$.

Since~$f'_v(P^A_v) = f_v(P^A_v) - f_v(\emptyset)$ for every~$v \in N$ the sum of the local scores is
\begin{align*}
\sum_{v \in N} f_v(P^A_v) &= \sum_{n \in N} (f'_v(P^A_v) + f_v(\emptyset))\\
& \geq t' + \sum_{n \in N} f_v(\emptyset) = t.
\end{align*}
To show that~$A$ is nice for~$I$, let~$f_v(P^A_v) \leq f_v(\emptyset)$. By the construction of~$\Fa$, this implies that~$f'_v(P^A_v)=0$ and therefore~$f'_v(P^A_v) = f'_v(\emptyset)$. Since~$A$ is nice for~$I'$ we conclude~$P^A_v= \emptyset$. Hence,~$A$ is a nice solution for~$I$. $\hfill \Box$
\end{proof}

\paragraph{Potential Parent Sets.} Given an instance~$I:=(N,\Fa,t,k)$ and some~$v \in N$, we denote the \emph{potential parent sets of~$v$} by~$\mathcal{P}_\Fa (v) := \{ P \subseteq N \setminus \{v\}: f_v(P)>0\} \cup \{\emptyset\}$, which are exactly the parent sets stored in~$\Fa$. If~$\Pi$ is monotone, we can assume by Proposition~\ref{Prop: Potential Parents} that in a solution~$A$ of~$I$, every vertex~$v$ has a parent set~$P^A_v \in \mathcal{P}_\Fa (v)$. An important measurement for the running times of our algorithms  is the maximum number of potential parent sets~\iflong{}$\delta_\Fa$ which is formally defined by \fi$\delta_\Fa := \max_{v \in N} |\mathcal{P}_\Fa (v)|$~\cite{OS13}. Given a vertex~$v \in N$, we can iterate over all potential parent sets of~$v$ and the vertices in these sets in~$\Oh(\delta_\Fa \cdot n)$ time.  

Another tool for designing algorithms for~BNSL problems is the superstructure~\cite{OS13}. \iflong Let~$N$ be a vertex set with local scores~$\Fa$. \fi The \emph{superstructure} of~$N$ and~$\Fa$ is the directed graph~$S_{\vec{\Fa}}=(N,A_{\Fa})$ with~$A_{\Fa}=\{(u,v) \mid \exists P \in \mathcal{P}_\Fa(v): u \in P\}$. \iflong Intuitively, there exists an arc~$(u,v) \in A_{\Fa}$ if and only if~$u$ is a potential parent of~$v$.  Given~$N$ and~$\Fa$, the superstructure~$S_{\vec{\Fa}}$ can be constructed in linear time. Throughout this work we let~$m := |A_{\Fa}|$ denote the number of arcs in the superstructure. Note that~$m < n^2$.\fi

 \subsection{Parameterized Complexity} In parameterized complexity~\cite{CFKLMPPS15} one measures the running time of algorithms depending on the total input size and a problem parameter.
 A \emph{parameterized problem} is a language~$L \subseteq \Sigma \times \mathds{N}_0$ over a finite alphapet~$\Sigma$. For an instance~$(I,k)$ of~$L$ we call~$k$ the \emph{parameter}.
 A parameterized problem~$L$ \emph{has an~XP-time algorithm} if for every instance~$(I,k)$ it can be decided in~$\Oh(|I|^{f(k)})$~time for a computable function~$f$ whether~$(I,k) \in L$. That is,  the problem is solvable in polynomial~time when the parameter is constant. A parameterized problem~$L$ is called~\emph{fixed-parameter tractable~(FPT)} if for every instance~$(I,k)$ it can be decided in $f(k) \cdot |I|^{\Oh(1)}$~time for a computable function~$f$ whether~$(I,k) \in L$. A \emph{problem kernelization} for a parameterized problem~$L$ is a polynomial-time preprocessing. Given an instance~$(I,k)$ of~$L$, a problem kernelization computes an equivalent instance~$(I',k')$ of~$L$ in polynomial time such that~$|I'|+k' \leq g(k)$ for some computable function~$g$. If~$g$ is a polynomial, then~$L$ \emph{admits a polynomial kernel}. Some fixed-parameter tractable problems are known to not admit a polynomial kernel unless~\badstuffhappens.


 A parameterized reduction maps an instance~$(I,k)$ of some parameterized problem~$L$ in FPT time to an equivalent instance~$(I',k')$ of a parameterized problem~$L'$ such that~$k' \leq f(k)$ for some computable funktion~$f$. If the reduction runs in~$|I|^{\Oh(1)}$ time and~$f$ is a polynomial, then the parameterized reduction is called~\emph{polynomial parameter transformation}. If there exists a parameterized reduction from a W[$i$]-hard parameterized problem~$L$ to a parameterized problem~$L'$, then~$L'$ is also~W[$i$]-hard for~$i \in \mathds{N}$. If a problem is W[$i$]-hard, then it is assumed to be \emph{fixed-parameter intractable}. If there exists a polynomial parameter transformation from a parameterized problem~$L$ to a parameterized problem~$L'$ and~$L$ does not admit a polynomial kernel unless~\badstuffhappens, then~$L'$ does not admit a polynomial kernel unless~\badstuffhappens~\cite{BTY11}.
 
\section{BNSL with Bounded Vertex Cover Number}
We first study the task of learning Bayesian network structures with a bounded vertex cover number. In the framework of constrained BNSL problems, these are the problems~\textsc{$(\Pi_0+v)$-Skeleton BNSL} and~\textsc{$(\Pi_0+v)$-Moral BNSL}, where~$\Pi_0$ is the class of edgeless graphs. Note that~$\Pi_0$ is monotone. Korhonen and Parviainen~\cite{KP15} analyzed the parameterized complexity of~\textsc{$(\Pi_0+v)$-Moral BNSL} for parameter~$k$. In their work, they provided an XP-time algorithm and proved W[1]-hardness. We adapt their approach to obtain an XP-time algorithm for~\textsc{$(\Pi_0+v)$-Skeleton BNSL}. Furthermore, we show a slightly stronger hardness result for both~problems.

\subsection{An XP-time Algorithm for Skeletons with Small Vertex Cover}

The XP-time algorithm for~\textsc{$(\Pi_0+v)$-Skeleton BNSL} follows the basic idea of the XP-time algorithm for~\textsc{$(\Pi_0+v)$-Moral BNSL}~\cite{KP15}: First, iterate over every possible choice of the vertex cover~$S$ and then split the arc set into two parts which are the arcs between~$S$ and the parents of~$S$ and the arcs between~$S$ and the children of~$S$. These two arc sets can be learned and combined independently.

However, we would like to point out that our algorithm for learning a network with bounded vertex cover number in the skeleton differs from the moralized version in one technical point. In the moralized graph, every vertex of a vertex cover~$S$ has at most one parent outside~$S$. For~\textsc{$(\Pi_0+v)$-Moral BNSL} this can be exploited to find the arcs between the vertices of~$S$ and their parents. However, this does not hold for the skeleton: Consider a star where all the arcs are directed towards a center. In this case, the central vertex forms a minimum vertex cover but the vertex has many parents. In the moralized graph, such star becomes a clique and the vertex cover number is large. To overcome this issue, we split the resulting network into three disjoint arc sets: The incoming arcs of vertices of~$S$, the incoming arcs of parents~$Q$ of vertices of~$S$, and the incoming arcs of the remaining vertices. 

In summary, the intuitive idea behind the algorithm is to find the vertex cover~$S$ and all parent sets of vertices in~$S$ via bruteforce. For each choice, we compute two further arc sets and combine them all to a solution of~\textsc{$(\Pi_0+v)$-Skeleton BNSL}. To find the incoming arcs of parents of~$S$, we adapt a dynamic programming algorithm for~\textsc{Vanilla-BNSL}~\cite{OM03,SM06}. With the next two lemmas, we formalize how our solution is built from disjoint arc sets.

\begin{lemma} \label{Lemma: VC Segmentation -> DAG}
Let~$(N,\Fa,t,k)$ be an instance of~\textsc{$(\Pi_0+v)$-Skeleton BNSL}, and let~$S$ and~$Q$ be disjoint subsets of~$N$. Furthermore, let there be arc sets~$B_1 \subseteq (Q \cup S) \times S$,~$B_2\subseteq S \times Q$, and~$B_3 \subseteq S \times (N \setminus (S \cup Q))$. If~$D':=(S \cup Q, B_1 \cup B_2)$ is a DAG where~$S$ is a vertex cover of~$\Sk(D')$, then
\begin{enumerate}
\item[a)] $D:=(N, A)$ with~$A:= B_1 \cup B_2 \cup B_3$ is a DAG,
\item[b)] $S$ is a vertex cover of~$\Sk(D)$, and
\item[c)] $\sum_{v \in N} f_v(P^A_v) = \sum_{v \in S} f_v(P^{B_1}_v) + \sum_{v \in Q} f_v(P^{B_2}_v) + \sum_{v \in N \setminus (S \cup Q)} f_v(P^{B_3}_v)$.
\end{enumerate}
\end{lemma}

\begin{proof}
Consider Statement~$a)$. Observe that if~$(v,w) \in B_3$, then~$w$ is a sink in~$D$. Together with the fact that~$D'$ is a DAG, this implies that~$D$ is a DAG.
Moreover, Statement~$b)$ holds, since every arc in~$A$ has at least one endpoint in~$S$.
For Statement~$c)$, observe that~$S$,~$Q$,~and~$(N \setminus (S \cup Q))$ form a partition of~$N$, and thus, every~$v \in N$ has incoming arcs from either~$B_1$,~$B_2$, or~$B_3$. $\hfill \Box$
\end{proof}

\begin{lemma} \label{Lemma: VC DAG -> Segmentation}
Let~$D:=(N,A)$ be a DAG such that~$S \subseteq N$ is a vertex cover in~$\Sk(D)$. Then, there exists a set~$Q \subseteq N \setminus S$ and arc sets~$B_1 \subseteq (Q \cup S) \times S$,~$B_2\subseteq S \times Q$, and~$B_3 \subseteq S \times (N \setminus (S \cup Q))$ that form a partition of~$A$. Moreover, every vertex in~$Q$ has a child in~$S$.
\end{lemma}

\begin{proof}
We set~$Q:= \{ v \in N \setminus S \mid v \text{ has a child in }S\}$. Then, every vertex in~$Q$ has a child in~$S$ by definition. Furthermore, we set~$B_1 := ((Q \cup S) \times S) \cap A$,~$B_2 := (S \times Q) \cap A$, and~$B_3 := (S \times (N \setminus (S \cup Q))) \cap A$.

Obviously,~$B_1 \cup B_2 \cup B_3 \subseteq A$, and the sets are pairwise disjoint, since~$S$,~$Q$, and~$N \setminus (S \cup Q)$ are disjoint subsets of~$N$. It remains to show that~$B_1 \cup B_2 \cup B_3 \supseteq A$. To this end, let~$(v,w) \in A$. If~$w \in S$, then~$v$ has a child in~$S$. Consequently,~$v \in S \cup Q$ and therefore,~$(v,w) \in B_1$. Otherwise, if~$w \not \in S$, then~$v \in S$, since~$S$ is a vertex cover of~$\Sk(D)$. Therefore,~$(v,w) \in B_2 \cup B_3$. $\hfill \Box$
\end{proof}

Intuitively, the algorithm works as follows: We iterate over all possible choices of~$S$,~$Q$, and~$B_1$. Then, for each such choice, we compute~$B_2$ and~$B_3$ that maximize the sum of local scores for~$A:=B_1 \cup B_2 \cup B_3$. In the following, we describe how to compute~$B_2$ when~$S$,~$Q$, and~$B_1$ are given. This step is the main difference between this algorithm and the XP-time algorithm for~\textsc{$(\Pi_0+v)$-Moral BNSL}~\cite{KP15}.

\begin{proposition} \label{Prop: Compute B_2}
Let~$I:=(N,\Fa,t,k)$ be an instance of~\textsc{$(\Pi_0+v)$-Skeleton BNSL}, and let~$S$ and~$Q$ be disjoint subsets of~$N$. Furthermore, let~$B_1 \subseteq (Q \cup S) \times S$ be an arc set such that~$(Q \cup S, B_1)$ is a DAG and every~$w \in Q$ has a child in~$S$. Then, we can compute an arc set~$B_2$ that maximizes~$\sum_{v \in Q} f_v(P^{B_2}_v)$ among all arc sets where~$(Q \cup S, B_1 \cup B_2)$ is a DAG and~$B_2 \subseteq S \times Q$ in~$2^{|S|} \cdot |I|^{\Oh(1)}$ time.
\end{proposition}

\begin{proof}
We describe a dynamic programming algorithm. 

\textit{Intuition.} Before we present the algorithm, we provide some intuition. Given a subset~$S' \subseteq S$ and the set~$Q' \subseteq Q$ containing parents of vertices in~$S$, we want to compute an arc-set~$B \subseteq S' \times Q'$ such that the sum of local scores for the arc-set~$B_1 \cup B$ is maximized. This is done by recursively choosing a vertex~$v \in S'$ that is a sink in the resulting DAG and letting all~$w \in Q'$ whose only child is~$v$ choose their best possible parent set in~$S' \setminus \{v\}$.

\textit{Algorithm.} To describe the algorithm, we introduce some notation. Given some~$w \in Q$, we let~$C^{B_1}_w$ denote the set of children of~$w$ in~$(S \cup Q, B_1)$. Note that~$C^{B_1}_w \subseteq S$ for all~$w \in Q$. Given a subset~$S' \subseteq S$, we let~$Q(S'):= Q \cap (\bigcup_{v \in S'} P^{B_1}_v)$ denote the set of parents of vertices in~$S'$ that belong to~$Q$, and~$D(S')$ denote the DAG with vertex set~$S' \cup Q(S')$ and arc set~$B_1 \cap ((S' \cup Q(S')) \times S')$. Furthermore, given~$S' \subseteq S$ and~$v \in S'$, we let
\begin{align*}
X(S',v) := \{w \in Q(S') \mid C^{B_1}_w \cap S' = \{v\}\}
\end{align*}
denote the vertices of~$Q(S')$ whose only child in~$S'$ is~$v$. Finally, given~$S' \subseteq S$ and~$w \in Q$, we define~$\widehat{f}_w(S'):= \max_{S'' \subseteq S'} f_w(S'')$ as the best possible score for a parent set of~$w$ containing only vertices from~$S'$. The values~$\widehat{f}_w(S')$ for all~$S' \subseteq S$ and~$w\in Q$ can be computed in overall~$2^{|S|} \cdot |I|^{\Oh(1)}$~time~\cite{OM03}.

The dynamic programming table~$T$ has entries of the type~$T[S']$ where~$S' \subseteq S$. Each entry stores the score of the best possible arc set~$B \subseteq S' \times Q(S')$ such that~$(S \cup Q, B_1 \cup B)$ is a DAG. For one-element sets~$\{v\} \subseteq S$, we set~$T[\{v\}]:= \sum_{w \in P_v^{B_1}} f_w(\emptyset)$. Note that, due to Proposition~\ref{Prop: Translation} we may assume that~$T[\{v\}]=0$. The recurrence to compute an entry for~$S'$ with~$|S'|>1$ is
\begin{align*}
T[S'] := \max_{\substack{v \in S' \\ v\text{ is a sink in }D(S')}} \left( T[S' \setminus \{v\}] + \sum_{w \in X(S',v)} \widehat{f}_w (S' \setminus \{v\}) \right).
\end{align*}

The score of the best possible arc set~$B_2 \subseteq S \times Q$ such that~$(S \cup Q, B_1 \cup B_2)$ is a DAG can be computed by evaluating~$T[S]$. The corresponding arc set can be found via traceback. The correctness proof is straightforward and thus omitted.

\textit{Running Time.} Recall that all values~$\widehat{f}_w(S')$ with~$S' \subseteq S$ and~$w \in Q$ can be computed in~$2^{|S|} \cdot |I|^{\Oh(1)}$ time. The dynamic programming table has~$2^{|S|}$ entries and each entry can be computed in~$|I|^{\Oh(1)}$~time. Thus, the overall running time is~$2^{|S|} \cdot |I|^{\Oh(1)}$ as claimed. $\hfill \Box$
\end{proof}

We now present the XP-time algorithm for~\textsc{$(\Pi_0+v)$-Skeleton BNSL}. This algorithm uses the algorithm behind Proposition~\ref{Prop: Compute B_2} as a subroutine.

\begin{theorem} \label{Theorem: VC XP Algo}
\textsc{$(\Pi_0+v)$-Skeleton BNSL} can be solved in~$(n\delta_\Fa)^k \cdot 2^k \cdot |I|^{\Oh(1)}$~time.
\end{theorem}

\begin{proof}
\textit{Algorithm.} Let~$I:= (N,\Fa,t,k)$ be an instance of~\textsc{$(\Pi_0+v)$-Skeleton BNSL}. The following algorithm decides whether~$I$ is a yes-instance or a no-instance: First, iterate over every possible choice of a vertex set~$S$ with~$|S| \leq k$ forming the vertex cover of the skeleton of the resulting network. For each choice of~$S$ iterate over every choice of potential parent sets for the vertices of~$S$. Let~$B_1$ be the corresponding arc set, and let~$Q$ be the set of parents of~$S$ in~$(N,B_1)$. For each choice of~$S$ and~$B_1$, do the following:
\begin{enumerate}
\item[•] Use the algorithm behind Proposition~\ref{Prop: Compute B_2} to compute an arc set~$B_2 \subseteq S \times Q$ that maximizes~$\sum_{v \in Q} f_v(P^{B_2}_v)$ among all arc sets where~$(Q \cup S, B_1 \cup B_2)$ is a DAG.
\item[•] For every~$v \in N\setminus (S \cup Q)$, compute a potential parent set that maximizes~$f_v(P)$ among all potential parent sets with~$P \subseteq S$. Let~$B_3 \subseteq S \times (N\setminus (S \cup Q))$ be the resulting arc set.
\item[•] If~$\sum_{v \in S} f_v(P^{B_1}_v) + \sum_{v \in Q} f_v(P^{B_2}_v) + \sum_{v \in N \setminus (S \cup Q)} f_v(P^{B_3}_v) \geq t$, then return~\emph{yes}.
\end{enumerate}
If for none of the choices of~$S$ and~$B_1$ the answer~\emph{yes} was returned, then return~\emph{no}.

\textit{Running Time.} First, we discuss the running time of the algorithm. Since~$|S| \leq k$, there are~$\Oh(n^k)$ choices for~$S$ and~$\Oh({\delta_\Fa}^k)$ choices for~$B_1$. For each such choice, the algorithm behind Proposition~\ref{Prop: Compute B_2} can be applied in~$2^k \cdot |I|^{\Oh(1)}$ time and the choice of the parent sets of vertices in~$N\setminus (S \cup Q)$ can be done in~$|I|^{\Oh(1)}$ time. This gives an overall running time of~$(n\delta_\Fa)^k \cdot 2^k \cdot |I|^{\Oh(1)}$ as claimed.

\textit{Correctness.} Second, we show that the algorithm returns~\emph{yes} if and only if~$I$ is a yes-instance.

$(\Rightarrow)$ Suppose the algorithm returns~\emph{yes}. Then, there exist disjoint subsets~$S$ and~$Q$ of~$N$, with~$|S| \leq k$ and arc sets~$B_1 \subseteq (S \cup Q) \times S$,~$B_2 \subseteq S \times Q$, and~$B_3\subseteq S \times N\setminus (S \cup Q)$ such that~$(Q \cup S, B_1 \cup B_2)$ is a DAG. Due to Lemma~\ref{Lemma: VC Segmentation -> DAG},~$D:=(N,A)$ with~$A:=B_1 \cup B_2 \cup B_3$ is a DAG and~$S$ is a vertex cover of~$\Sk(D)$. Moreover, $\sum_{v \in N} f_v(P^A_v) \geq t$ and therefore,~$I$ is a yes-instance.

$(\Leftarrow)$ Let~$I$ be a yes-instance. Then, there exists an~$(N,\Fa,t)$-valid arc set~$A$ such that the skeleton of~$D:=(N,A)$ has a vertex cover~$S$ of size at most~$k$. By Lemma~\ref{Lemma: VC DAG -> Segmentation}, there exists a set~$Q\subseteq N \setminus S$ and arc sets~$B_1 \subseteq (S \cup Q) \times S$,~$B_2 \subseteq S \times Q$, and~$B_3\subseteq S \times N\setminus (S \cup Q)$ that form a partition of~$A$. Since the algorithm iterates over all choices of~$S$ and~$B_1$ with~$|S| \leq k$, it considers~$S$ and~$B_1$ at some point. For this choice of~$S$ and~$B_1$ the algorithm then computes an arc set~$B'_2 \subseteq S \times Q$ with
\begin{align*}
\sum_{v \in Q} f_v(P^{B'_2}_v) \geq  \sum_{v \in Q} f_v(P^{B_2}_v)
\end{align*}
and an arc set~$B'_3\subseteq S \times N\setminus (S \cup Q)$ with
\begin{align*}
\sum_{v \in N\setminus (S \cup Q)} f_v(P^{B'_3}_v) \geq  \sum_{v \in N\setminus (S \cup Q)} f_v(P^{B_3}_v).
\end{align*}
Then, since the sum of the local scores under~$A$ is at least~$t$, the algorithm returns~\emph{yes}. $\hfill \Box$
\end{proof}

\subsection{W[2]-hardness for Skeletons with Small Vertex Cover}

We complement the XP-time algorithm from the previous subsection by proving~W[2]-hardness of~\textsc{$(\Pi_0+v)$-Skeleton BNSL}. Thus,~\textsc{$(\Pi_0+v)$-Skeleton BNSL} is not FPT for parameter~$k$ unless~$\text{W[2]}=\text{FPT}$. 
We show that the hardness also holds for the task of learning a Bayesian network where the skeleton has a a bounded dissociation number. Formally, this is~\textsc{$(\Pi_1+v)$-Skeleton BNSL}, with~$\Pi_1 := \{G \mid G \text{ has maximum degree }1\}$. Observe that~$\Pi_1$ is monotone.

\begin{theorem} \label{Theorem: Bounded-VC W[2]-h}
Let~$\Pi \in \{\Pi_0, \Pi_1\}$. Then, \textsc{$(\Pi+v)$-Skeleton BNSL} is W[2]-hard for~$k$ even when the superstructure is a DAG, the maximum parent set size is 1, and every local score is either 1 or 0.
\end{theorem}

\begin{proof}
We give a parameterized reduction from~\textsc{Set Cover}. In~\textsc{Set Cover}, one is given a universe~$U$, a family~$\mathcal{X}$ of subsets of~$U$, and an integer~$\ell$. The question is, whether there exists a a subfamily~$\mathcal{X}' \subseteq \mathcal{X}$ with~$|\mathcal{X}'| \leq \ell$ that covers~$U$. That is, every~$u \in U$ is contained in some set of~$\mathcal{X}'$. \textsc{Set Cover} is~W[2]-hard when parameterized by~$\ell$~\cite{CFKLMPPS15}. We first describe a parameterized reduction from~\textsc{Set Cover} to~\textsc{$(\Pi_0+v)$-Skeleton BNSL} and afterwards, we describe how this construction can be modified to obtain W[2]-hardness for~\textsc{$(\Pi_1+v)$-Skeleton BNSL}.

\textit{Construction.} Let~$(U, \mathcal{X}, \ell)$ be an instance of~\textsc{Set Cover}. We describe how to construct an equivalent instance~$I:=(N,\Fa,t,k)$ with~$k=\ell$. First, we set~$N:=U \cup \{v_X \mid X \in \mathcal{X}\}$. Next, we define the local scores~$\Fa$. All local scores are either 1 or 0. For every~$u \in U$ we set~$f_u(P)=1$ if and only if~$P=\{v_X\}$ for some~$X \in \mathcal{X}$ that contains~$u$. Furthermore, for every~$v \in \{v_X \mid X \in \mathcal{X}\}$, we set~$f_v(P)=0$ for every~$P$. To finish the construction, we set~$k:=\ell$ and~$t:=|U|$.

Observe that for every arc~$(u,v)$ of the superstructure, we have~$u \in U$ and~$v \in \{v_X \mid X \in \mathcal{X}\}$. Consequently, the super structure is a DAG. Furthermore, by the construction of~$\Fa$, the maximum parent set size is~$1$.

\textit{Intuition.} Before we show the correctness, we provide some intuition. To obtain a score of~$t=|U|$, every vertex in~$U$ has to choose one parent vertex. The chosen parent vertices correspond to the subfamily~$\mathcal{X}' \subseteq \mathcal{X}$ that covers~$U$. The vertex cover constraint on the network ensures that~$\mathcal{X}'$ has size at most~$k$.

\textit{Correctness.} We show that~$(U, \mathcal{X}, \ell)$ is a yes-instance of \textsc{Set Cover} if and only if~$I$ is a yes-instance of~\textsc{$(\Pi_0+v)$-Skeleton BNSL}.

$(\Rightarrow)$ Let~$\mathcal{X}' \subseteq \mathcal{X}$ be a subfamily of size at most~$k$ that covers~$U$. Then, for every~$u \in U$, there exists some set~$X^u \in \mathcal{X}'$ that contains~$u$. We define~$A:=\{(v_{X^u},u) \mid u \in U\}$ and show that~$A$ is a solution of~$I$.

Consider the skeleton~$\Sk(N,A)$. Each connected component of~$\Sk(N,A)$ is either an isolated vertex or a star consisting of a central vertex from~$\{v_X \mid X \in \mathcal{X}'\}$ and leaf vertices from~$U$. Thus,~$(N,A)$ is a~DAG and~$\{v_X \mid X \in \mathcal{X}'\}$ is a vertex cover of the skeleton. Thus,~$\Sk(N,A) \in \Pi_0+kv$, since~$|\mathcal{X}'| \leq k$. Moreover, observe that~$f_u(P^A_u)=1$ for every~$u \in U$. Therefore,~$A$ is~$(N,\Fa,t)$-valid.

$(\Leftarrow)$ Conversely, let~$A$ be an~$(N,\Fa,t)$-valid arc set such that~$\Sk(N,A)$ has a vertex cover of size at most~$k$. Then, since~$t=|U|$, we have~$f_u(P^A_u)=1$ for every~$u \in U$. Thus, for every~$u \in U$ we have~$P^A_u=\{v_X\}$ for some~$X \in \mathcal{X}$ containing~$u$. We define~$\mathcal{X}' := \{ X \in \mathcal{X} \mid P^A_u =\{v_X\} \text{ for some }u \in U \}$.

We first show that~$\mathcal{X}'$ covers~$U$. Let~$u \in U$. Then,~$P^A_u= \{v_X\}$ for some~$X$ containing~$u$ and therefore~$X \in \mathcal{X}'$. Thus,~$\mathcal{X}'$ covers~$U$. It remains to show that~$|\mathcal{X}'| \leq k$. Assume towards a contradiction that~$|\mathcal{X}'| > k$. Then, there exist pairwise distinct vertices~$u_1, \dots, u_{k+1}$ in~$U$ and~$v_1, \dots, v_{k+1}$ in~$\{v_X \mid X \in \mathcal{X}'\}$ such that~$(v_i,u_i) \in A$ for~$i\in \{1, \dots, k+1\}$. Then, the edges~$\{v_i,v_i\}$ form a matching of size~$k+1$ in~$\Sk(N,A)$. This contradicts the fact that~$\Sk(N,A)$ has a vertex cover of size at most~$k$.

\textit{BNSL with bounded Dissociation Number.} We now explain how to modify the construction described above, to obtain W[2]-hardness for~\textsc{$(\Pi_1+v)$-Skeleton BNSL} when parameterized by~$k$.

In the construction, we set~$N:= U \cup \{v_X \mid X \in X\} \cup \{w_X \mid X \in X\}$. As in the construction described above, for~$u\in U$ we set~$f_u(P) := 1$ if and only if~$P=\{v_X\}$ for some~$X \in \mathcal{X}$ containing~$u$, and for~$v \in \{v_X \mid X \in X\}$ we set~$f_v(P):=0$ for every~$P$. Additionally, for every~$w_X$, we set~$f_{w_X}(P):=1$ if and only if~$P=\{v_X\}$. Furthermore, we set~$k:=\ell$ and~$t:=|U|+|\mathcal{X}|$

$(\Rightarrow)$ Let~$\mathcal{X}' \subseteq \mathcal{X}$ be a subfamily with~$|\mathcal{X}'| \leq k$ that covers~$U$. We set~$A:= \{(v_{X^u},u) \mid u \in U\} \cup \{(v_X, w_X) \mid X \in \mathcal{X}\}$. Then,~$(N,A)$ is a DAG and the sum of the local scores is~$t$. Furthermore, the connected components of~$\Sk(N,A)$ are isolated edges or disjoint stars with central vertex in~$\{v_X \mid X \in \mathcal{X}'\}$. Then, $|\mathcal{X}'| \leq k$ implies~$\Sk(N,A) \in \Pi_1 +kv$.

$(\Leftarrow)$ Let~$A$ be a solution of~$I$. Again, we define~$\mathcal{X}' := \{ X \in \mathcal{X} \mid P^A_u =\{v_X\} \text{ for some }u \in U \}$, which  covers~$U$ by the same arguments as above. Note that the skeleton of~$(N,A)$ contains an edge~$\{v_X, w_X \}$ for every~$X \in \mathcal{X}$, since the sum of local scores under~$A$ is at least~$t$. Then, assuming~$|\mathcal{X}'| > k$ implies that there exist~$k+1$ vertex disjoint sets~$\{u,v_X,w_X\}$ where~$v_X$  is adjacent with~$u$ and~$w_X$ in~$\Sk(N,A)$. This contradicts the fact that~$\Sk(N,A)$ has a dissociation set of size at most~$k$. $\hfill \Box$
\end{proof}

Observe that for a DAG~$D:=(N,A)$ where each vertex has at most one parent, the skeleton~$\Sk(D)$ and the moralized graph~$\Mo(D)$ are the same. Thus, Theorem~\ref{Theorem: Bounded-VC W[2]-h} also implies W[2]-hardness if the sparsity constraints are posed on the moralized graph. Note that a~W[1]-hardness for~\textsc{$(\Pi_0+v)$-Moral BNSL} when parameterized by~$k$ has been shown~\cite{KP15}. We now obtain a slightly stronger hardness result with an additional restriction on the maximum parent set size.

\begin{corollary} \label{Cor: Bounded-VC W[2]-h}
Let~$\Pi \in \{\Pi_0, \Pi_1\}$. Then, \textsc{$(\Pi+v)$-Moral BNSL} is W[2]-hard for~$k$ even when the superstructure is a DAG, the maximum parent set size is 1, and every local score is either 1 or 0.
\end{corollary}

By Corollary~\ref{Cor: Bounded-VC W[2]-h} it is presumably not possible that~\textsc{$(\Pi_1+v)$-Moral BNSL} is FPT for~$k$. 

Consider networks where the maximum parent set size is 1. These networks are also known as \emph{branchings}. Learning a branching without further restrictions can be done in polynomial time~\cite{CL68,GKLOS15}. Due to Theorem~\ref{Theorem: Bounded-VC W[2]-h} and Corollary~\ref{Cor: Bounded-VC W[2]-h}, there is presumably no such polynomial-time algorithm if we add a sparsity constraint on the vertex cover size. Thus, the task to learn a branching is an example where it is harder to learn a more restricted~network.

\section{BNSL with Bounded Dissociation Number} \label{Section: DissNo}
In this section we provide an algorithm for \textsc{$(\Pi_1+v)$-Moral BNSL}, that is, for Bayesian network learning where the moralized graph has dissociation number at most~$k$. By the results above, an FPT algorithm for~$k$ is unlikely. We show that it can be solved in~XP-time when parameterized by~$k$. As detailed in the introduction, this shows that we can find optimal networks for a class of moral graphs that is larger than the ones with bounded vertex cover number, while maintaining the highly desirable property that the treewidth is bounded. In fact, graphs with dissociation number at most~$k$ have treewidth at most~$k+1$ and thus the Bayesian inference task can be preformed efficiently if~$k$ is small~\cite{D09}.

Before we describe the main idea of the algorithm, we provide the following simple observation \iflong about Bayesian networks whose moralized graph has a bounded dissociation number\fi.

\begin{proposition} \label{Prop: <= 2|S| ancestors}
Let~$D=(N,A)$ be a DAG and~$S \subseteq N$ be a dissociation set of~$\Mo(D)$. Then, at most~$2|S|$ vertices in~$N\setminus S$ have descendants in~$S$.
\end{proposition}

\iflong
\begin{proof}
Let~$v \in S$. We call a vertex~$w \in N \setminus S$ is an~\emph{external ancestor of~$v$} if there exists a path~$(w, w_1, \dots, w_\ell, v)$ in~$D$ such that $w_i \in N \setminus S$ for all~$i \in \{1, \dots, \ell\}$. We show that every vertex in~$S$ has at most~two external ancestors.

First, assume that~$v$ has three distinct parents~$w_1$, $w_2$, and~$w_3$ outside~$S$. Then, there are moral edges~$\{w_1, w_2\}$, $\{w_2, w_3\}$, and~$\{w_3, w_1\}$ forming a triangle outside~$S$ in~$\Mo(D)$. This contradicts the fact that~$S$ is a dissociation set of~$\Mo(D)$. Hence, every~$v \in S$ has at most two parents outside~$S$. Next, consider the following cases.

\textbf{Case 1:~$|P^A_v \setminus S| = 0$.} Then, $v$ has no external ancestors and nothing more needs to be shown.

\textbf{Case 2:~$|P^A_v \setminus S| = 1$.} Then, let~$P^A_v \setminus S = \{w\}$. Since~$S$ is a dissociation set of~$\Mo(D)$ it holds that~$\deg_{\Mo(D)-S}(w) \leq 1$. Hence,~$w$ has at most one parent~$w'$ outside~$S$. Moreover, since~$\deg_{\Mo(D)-S}(w') \leq 1$, the vertex~$w'$ has no parent in~$N \setminus S$. Therefore,~$v$ has at most two external ancestors.

\textbf{Case 3:~$|P^A_v \setminus S| = 2$.} Then, let~$P^A_v \setminus S = \{w_1, w_2\}$. Note that~$\{w_1, w_2\}$ is a moral edge in~$\Mo(D)$. Then, since~$\deg_{\Mo(D)-S}(w_1) \leq 1$ and~$\deg_{\Mo(D)-S}(w_2) \leq 1$, the vertices~$w_1$ and~$w_2$ do not have parents in~$N \setminus S$. Therefore,~$v$ has exactly two external ancestors. $\hfill \Box$
\end{proof}
\fi

The main idea of the algorithm for~\textsc{$(\Pi_1+v)$-Moral BNSL} presented in this work is closely related to XP-algorithms for \textsc{$(\Pi_0+v)$-Moral BNSL}~\cite{KP15} and \textsc{$(\Pi_0+v)$-Skeleton BNSL} (Theorem~\ref{Theorem: VC XP Algo}): If we know which vertices form the dissociation set~$S$ and the set~$Q$ of vertices that are the ancestors of~$S$, the arcs of the network can be found efficiently. Roughly speaking, the steps of the algorithm are to iterate over every possible choice of~$S$ and~$Q$ and then find the arc set of the resulting network respecting this choice. Finding the arc set can then be done in two steps: First, we find all the arcs between the vertices of~$S \cup Q$ and afterwards, we find the remaining arcs of the network. Even though the basic idea of the algorithm is similar to algorithms for BNSL with bounded vertex cover number, several obstacles occur when considering~$\Pi_1$ instead of~$\Pi_0$. 

First, the arcs between~$S \cup Q$ and the remaining arcs of the DAG cannot be computed independently, since there might be arcs between vertices of~$Q$ and~$N \setminus (Q \cup S)$. See Figure~\ref{Figure: ancestor extension and Suitable Arc Set} for an example of a  DAG~$D$ whose moralized graph has a dissociation set~$S$. We overcome this obstacle by partitioning~$Q$ into two sets~$Q_0$ and~$Q_1$ and by considering arc sets~$A_Q \subseteq (S \cup Q) \times (S \cup Q)$ that respect a specific constraint regarding this partition.

Second, the vertices in~$N \setminus (S \cup Q)$ cannot choose their parent sets greedily from~$S$, since they may also choose one parent from~$N \setminus S$. Thus, we need a new technique to find this part of the network. To overcome this obstacle, we define the problem \textsc{Basement Learning} and show that it can be solved in polynomial time.

This section is organized as follows: In Section~\ref{Subsection: Introduce Attic and Basement}, we introduce the terms of \emph{attic arc sets} and \emph{basement arc sets} which form the parts of the arc set that we later combine to a solution. In Sections~\ref{Subsection: Find Attic}, we describe how to find the attic arc set and in Section~\ref{Subsection: Find Basement}, we describe how to find the basement arc set. Finally, in Section~\ref{Subsection: Bounded Diss-No Algo}, we combine the previous results and describe how to solve~\textsc{$(\Pi_1+v)$-Moral BNSL} in XP-time. 

We end this section by showing another hardness result for~\textsc{$(\Pi_1+v)$-Moral BNSL}. Note that due to Corollary~\ref{Cor: Bounded-VC W[2]-h} it is unlikely that~\textsc{$(\Pi_1+v)$-Moral BNSL} is FPT for~$k$. In Section~\ref{Subsection: DissNoBNSL W[1]-h for k+t} we show that even for parameterization by~$k+t+\delta_\Fa+p$, where~$p$ denotes the maximum parent set size, it is unlikely to obtain an FPT algorithm for~\textsc{$(\Pi_1+v)$-Moral BNSL}.

\iflong
\subsection{Attic Arc Sets and Basement Arc Sets} \label{Subsection: Introduce Attic and Basement}
\fi

In this subsection we formally define \emph{attic arc sets} and \emph{basement arc sets}. As mentioned above, these are the two parts of the resulting network that our algorithm finds when the vertices of the dissociation set and their ancestor vertices are known. The intuitive idea behind the names~\emph{attic arc set} and~\emph{basement arc set} is that the dissociation set~$S$ forms the center of the network, the arcs between~$S$ and the ancestors of~$S$ form the upper part of the network (attic) and the remaining arcs from the lower part (basement) of the network. Figure~\ref{Figure: ancestor extension and Suitable Arc Set} shows a DAG~$D$ where the arcs are decomposed into an attic arc set and a basement arc set. 

Throughout this section, we let~$S$ denote the set of the vertices that form the dissociation set and we let~$Q$ denote the set of their ancestors. Furthermore, we assume that~$Q$ is partitioned into two sets~$Q_0$ and~$Q_1$. Intuitively, in the moralized graph of the resulting network, the vertices in~$Q_0$ have no neighbors in~$Q$ and the vertices in~$Q_1$ may have one neighbor in~$Q$.

\begin{figure}
\begin{center}
\begin{tikzpicture}[scale=0.85,yscale=0.7]
\tikzstyle{knoten}=[circle,fill=white,draw=black,minimum size=5pt,inner sep=0pt]
\tikzstyle{bez}=[inner sep=0pt]

\draw[rounded corners] (2, 0) rectangle (6, 1.4) {};
\node[knoten] (v1) at (3,0.7) {};
\node[knoten] (v2) at (4,0.7) {};
\node[knoten] (v3) at (5,0.3) {};
\node[knoten] (v4) at (5,1.1) {};

\draw[rounded corners] (-1, 0) rectangle (1.9, 1.4) {};
\node[knoten] (v5) at (0,0.7) {};
\node[knoten] (v6) at (1,0.7) {};

\draw (4,-1.5) ellipse (2cm and 0.7cm);
\node[knoten] (v7) at (2.4,-1.5) {};
\node[knoten] (v8) at (3,-1.5) {};
\node[knoten] (v9) at (3.6,-1.5) {};
\node[knoten] (v10) at (4.2,-1.5) {};
\node[knoten] (v11) at (4.8,-1.5) {};
\node[knoten] (v12) at (5.4,-1.5) {};

\draw[rounded corners] (-1.5, -4) rectangle (6.5, -3) {};
\node[knoten] (r1) at (-1,-3.5) {};
\node[knoten] (r2) at (0,-3.5) {};
\node[knoten] (r3) at (1,-3.5) {};
\node[knoten] (r4) at (2,-3.5) {};
\node[knoten] (r5) at (3,-3.5) {};
\node[knoten] (r6) at (4,-3.5) {};
\node[knoten] (r7) at (5,-3.5) {};
\node[knoten] (r8) at (6,-3.5) {};

\draw[->]  (v4) to (v3);
\draw[->, bend right=10]  (v1) to (v10);
\draw[->]  (v2) to (v10);
\draw[->, bend right=10]  (v11) to (v2);
\draw[->]  (v3) to (v12);
\draw[->, bend right]  (v5) to (v7);
\draw[->]  (v6) to (v8);
\draw[->]  (v10) to (v9);
\draw[->]  (v7) to (v8);

\draw[->, line width=1pt]  (v5) to (r1);
\draw[->, line width=1pt]   (r3) to (r2);
\draw[->, line width=1pt]   (v8) to (r3);
\draw[->, line width=1pt]   (v7) to (r2);
\draw[->, line width=1pt]   (v8) to (r4);
\draw[->, line width=1pt]   (v11) to (r7);
\draw[->, line width=1pt, bend right=10]   (v11) to (r8);
\draw[->, line width=1pt, bend left=10]   (v12) to (r8);
\draw[->, line width=1pt]   (r7) to (r6);

\draw[-, very thick, densely dotted]  (v1) to (v2);
\draw[-, very thick, densely dotted,  bend left=10]  (r3) to (v7);
\draw[-, very thick, densely dotted]   (v7) to (v6);
\draw[-, very thick, densely dotted]   (v11) to (v12);

\node[bez] (Q) at (7,0.7) {$Q$};
\node[bez] (Q0) at (-0.6,1.1) {$Q_0$};
\node[bez] (Q0) at (2.4,1.1) {$Q_1$};

\node[bez] (S) at (7,-1.5) {$S$};

\node[bez] (R) at (7,-3.5) {$R$};
\end{tikzpicture}
\end{center}
\caption{\small{A DAG~$D$ whose moralized graph has a dissociation set~$S$. The arc set of~$D$ is decomposed into an attic arc set~$A_Q$ and a basement arc set~$A_R$. The thin arrows correspond to the arcs of~$A_Q$ and the thick arrows correspond to the arcs of~$A_R$. The dotted edges are the moral edges.}}\label{Figure: ancestor extension and Suitable Arc Set}
\end{figure}
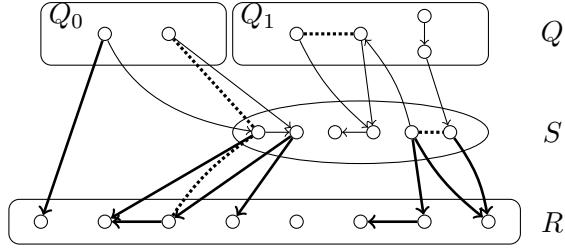

\begin{definition} \label{Def: ancestor extension}
Let~$N$ be a vertex set and let~$S$, $Q_0$, and~$Q_1$ be disjoint subsets of~$N$. An arc set~$A_Q$ is called \emph{attic arc set of~$S$, $Q_0$, and~$Q_1$}, if
\begin{enumerate}
\item[a)] $D_Q := (S \cup Q_0 \cup Q_1, A_Q)$ is a DAG, 
\item[b)] in the moralized graph~$\Mo(D_Q)$, no vertex of~$Q_0$ has neighbors outside~$S$, and every vertex of~$Q_1$ has at most one neighbor outside~$S$.
\end{enumerate}
\end{definition}

If~$S$,~$Q_0$, and~$Q_1$ are clear from the context we may refer to~$A_Q$ as \emph{attic arc set}. Throughout this section we use the following notation as a shorthand for some vertex sets: Given a vertex set~$N$ and disjoint subsets~$S$,~$Q_0$, and~$Q_1$ of~$N$, we let~$Q:=Q_0 \cup Q_1$, and we let~$R:= N \setminus (S \cup Q)$ denote the remaining vertices of~$N$. We next define basement arc sets. 

\begin{definition} \label{Def: suitable arc sets}
Let~$N$ be a vertex set and let~$S$,~$Q_0$, and~$Q_1$ be disjoint subsets of~$N$. An arc set~$A_R \subseteq (S \cup Q_0 \cup R) \times R$ is called \emph{basement arc set for~$S$,~$Q_0$, and~$Q_1$}~if~$A_R$ contains no self-loops and every~$w \in Q_0 \cup R$ has at most one incident arc in~$A_R \cap ((R\cup Q_0) \times R)$. 
\end{definition}

If~$S$,~$Q_0$, and~$Q_1$ are clear from the context we may refer to~$A_R$ as \emph{basement arc set}. The idea is that basement arc sets and attic arc sets can be combined to a solution of~\textsc{$(\Pi_1+v)$-Moral BNSL} and that a solution can be splitted into an attic arc set and a basement arc set. With the next two lemmas, we formalize this idea. First, an attic arc set and a basement arc set can be combined to a DAG where~$S$ is a dissociation set of the moralized graph. 

\begin{lemma} \label{Lemma: segmentation -> DAG}
Let~$N$ be a vertex set and let~$S$,~$Q_0$ and~$Q_1$ be disjoint subsets of~$N$. Furthermore, let~$A_Q$ be an attic arc set, and let~$A_R$ be a basement arc set. \iflong It then holds that
\begin{enumerate}
\item[1.] $D:=(N,A_Q \cup A_R)$ is a DAG, and
\item[2.] $S$ is a dissociation set of~$\Mo(D)$.
\end{enumerate}
\fi
\end{lemma}

\iflong
\begin{proof}
We first show that~$D$ is a~DAG. Assume towards a contradiction that there is a directed cycle in~$D$. Since~$A_Q$ is an attic arc set we conclude from Definition~\ref{Def: ancestor extension}~$a)$ that there is no directed cycle in~$(N,A_Q)$. Hence, the cycle contains an edge~$(v,w) \in A_R$. Note that~$w \in R$ and there exists an outgoing edge~$(w,w') \in A_Q \cup A_R$ that is also part of the cycle. Since no edge in~$A_Q$ is incident with vertices of~$R$ we conclude~$(w, w') \in A_R$ and therefore~$w' \in R$. Note that~$w' \neq w$, since~$A_R$ contains no self-loops. Since~$(w,w')$ is part of the directed cycle, there exists an edge~$(w',w'') \in A_R$ with~$w'' \in R$. Then,~$w'$ is incident with two arcs in~$A_R \cap ((R \cup Q_0) \times R)$ which is a contradiction to the fact that~$A_R$ is a basement arc set. Consequently, there is no directed cycle in~$D$.

It remains to show that~$S$ is a dissociation set of~$\Mo(D)$. That is, we show that every vertex~$v$ has degree at most 1 in~$G:=\Mo(D)-S$. 

If~$v \in Q_1$, then $v$ has degree at most~one in~$\Mo(D_Q)-S$. Since no arc in~$A_R$ is incident with~$v$, we conclude~$\deg_G(v) =1$. 

Otherwise,~$v \in Q_0 \cup R$. Then, there is no arc in~$A_Q$ connecting~$v$ with a vertex in~$N \setminus S$. Moreover, by Definition~\ref{Def: suitable arc sets}, there is at most one arc in~$A_R \cap ((R \cup Q_0) \times R)$ that is incident with~$v$. To prove~$\deg_G(v) \leq 1$ it remains to show that there is no moral edge of~$\Mo(D)$ connecting~$v$ with some other vertex in~$N \setminus S$. Assume towards a contradiction that there exists some~$v' \in N \setminus S$ such that~$v$ and~$v'$ have a common child~$w$. 

If~$w \in S \cup Q$, then~$v \in Q_0$ and~$v' \in Q$. Consequently,~$\{v,v'\}$ is a moral edge in~$\Mo(D_Q)$ which contradicts the fact that vertices in~$Q_0$ have degree 0 in~$\Mo(D_Q)-S$. Hence, we conclude~$w \in R$ and therefore~$(v,w), (v',w) \in A_R$. Then,~$w$ has two incident arcs in~$A_R \cap ( (R \cup Q_0) \times R)$ which contradicts the fact that~$A_R$ is a basement arc set. Hence,~$\deg_G(v) \leq 1$. $\hfill \Box$
\end{proof}
\fi

Next, we show that conversely, the arc set of every DAG whose moralized graph has a dissociation set~$S$ can be partitioned into the an attic arc set and a basement arc set for some~$Q_0$ and~$Q_1$.

\begin{lemma} \label{Lemma: DAG -> segmentation}
Let~$D=(N,A)$ be a DAG and let~$S \subseteq N$ be a dissociation set of~$\Mo(D)$. Furthermore, let~$Q:= \{v \in N \setminus S \mid v \text{ has at least one descendant in }S \}$ and let~$Q$ be partitioned into
\begin{align*}
Q_0 &:= \{v \in Q \mid v \text{ has degree } 0 \text{ in }\Mo(S \cup Q, A_Q)-S\} \text{, and}\\
Q_1 &:= Q \setminus Q_0.
\end{align*}
Then, $A_Q:= ((S \cup Q)\times (S \cup Q) ) \cap A$ is an attic arc set and~$A \setminus A_Q$ is a basement arc set. Moreover,~$|Q| \leq 2 |S|$.
\end{lemma}

\iflong
\begin{proof}
Note that Proposition~\ref{Prop: <= 2|S| ancestors} implies~$|Q| \leq 2|S|$. We first show that Properties~$a)$ and~$b)$ from Definition~\ref{Def: ancestor extension} hold for~$A_Q$. Since~$D=(N,A)$ is a DAG,~$S \cup Q \subseteq N$, and~$A_Q \subseteq A$, it clearly holds that~$D_Q$ is a DAG and therefore Property~$a)$ holds. Consider Property~$b)$. By the definition of~$Q_0$, no vertex in~$Q_0$ has neighbors in~$Q$. Furthermore, since~$S$ is a dissociation set of~$\Mo(D)$, every vertex in~$Q_1$ has at most one neighbor in~$Q$.

It remains to show that~$A \setminus A_Q$ is a basement arc set. To this end, we first show~$A \setminus A_Q \subseteq (S \cup Q_0 \cup R) \times R$. Assume towards a contradiction that~$A \setminus A_Q \not \subseteq (S \cup Q_0 \cup R) \times R$. Consider the following cases.

\textbf{Case 1:} There exists an arc~$(v,w) \in A \setminus A_Q$ with~$w \not \in R$\textbf{.} Then,~$w \in S \cup Q$ and therefore,~$v$ is an ancestor of~$S$. Hence,~$(v,w) \in A_Q$ which contradicts the choice of~$(v,w)$.

\textbf{Case 2:} There exists an arc~$(v,w) \in A \setminus A_Q$ with~$v \in Q_1$\textbf{.} From the previous case we know~$w \in R$. Since~$v$ has degree 1 in~$\Mo(N,A_Q)-S$ and an incident arc to some vertex in~$R$ we conclude~$\deg_{\Mo(D)-S}(v)\geq 2$ which contradicts the fact that~$S$ is a dissociation set of~$\Mo(D)$. Since Cases 1 and 2 are contradictory, we have~$A \setminus A_Q \subseteq (S \cup Q_0 \cup R) \times R$.

Finally, we show that Definition~\ref{Def: suitable arc sets} holds for~$A \setminus A_Q$. Since~$D$ is a DAG we conclude that~$A \setminus A_Q$ contains no self-loops. Moreover, since~$S$ is a dissociation set of~$\Mo(D)$ we conclude that every~$w \in (Q_0 \cup R)$ has at most one incident edge in~$(A \setminus A_Q) \cap ((R \cup Q_0) \times R)$. $\hfill \Box$
\end{proof}
\fi

In general, if we consider a union~$A_1 \cup A_2$ of two disjoint arc-sets, one vertex~$v$ may have incoming arcs from~$A_1$ and~$A_2$. Thus, for the local scores we may have~$f_v(P^{A_1}_v) \neq f_v(P^{A_1 \cup A_2}_v)$. Given an attic arc set~$A_Q$ and a basement arc set~$A_R$, all arcs in~$A_R$ have endpoints in~$R$ and all arcs in~$A_Q$ have endpoints in~$Q \cup S$. Since~$Q \cup S$ and~$R$ are disjoint, for every vertex~$v $ either all incoming arcs are in~$A_Q$ or in~$A_R$. Thus, the local scores under~$A_Q \cup A_R$ can be decomposed as follows.

\begin{lemma} \label{Lemma: Decompose Score}
Let~$(N,\Fa,t,k)$ be an instance of~\textsc{$(\Pi_1+v)$-Moral BNSL} and let~$S$,~$Q_0$, and~$Q_1$ be disjoint subsets of~$N$. Furthermore, let~$A_Q$ be an attic arc set and let~$A_R$ be a basement arc set. Then, score of~$A:= A_Q \cup A_R$ under~$\Fa$~is
\begin{align*}
\sum_{v \in N} f_v(P_v^{A}) = \sum_{v \in S \cup Q} f_v(P_v^{A_Q}) + \sum_{v \in R} f_v(P_v^{A_R}).
\end{align*}
\end{lemma}

\subsection{Finding the Attic Arc Set} \label{Subsection: Find Attic}

Recall that the intuitive idea of the XP-time algorithm is to iterate over all possible vertices that may form the dissociation set and their possible ancestors. Then, for each choice we find an attic arc set and a basement arc set. In this subsection, we present an algorithm to efficiently compute the attic arc set when~$S$,~$Q_0$, and~$Q_1$ are given.

Let~$I:=(N,\Fa,t,k)$ be an instance of~\textsc{$(\Pi_1+v)$-Moral BNSL} and let~$S$, $Q_0$, and~$Q_1$ be disjoint subsets of~$N$. An attic arc set~$A_Q$ is called~\emph{optimal}, if~$\sum_{v \in S \cup Q_0 \cup Q_1} f_v(P_v^{A_Q})$ is maximal among all attic arc sets for~$S$,~$Q_0$, and~$Q_1$.

Let~$\lambda:=|S \cup Q_0 \cup Q_1|$. Observe that, by iterating over every possible set of arcs between the vertices in~$S \cup Q_0 \cup Q_1$, one can enumerate all possible~$A_Q$ in~$2^{\Oh(\lambda^2)} \cdot |I|^{\Oh(1)}$~time. Alternatively, by iterating over all possible parent sets of the vertices of~$S \cup Q_0 \cup Q_1$, one can enumerate all possible~$A_Q$ in~${\delta_\Fa}^{\lambda} \cdot |I|^{\Oh(1)}$~time. However, this might be expensive, since~$\delta_\Fa$ can be exponentially large in the number of vertices. We show that an optimal attic arc set can be computed in~$\lambda^{\Oh(\lambda)} \cdot |I|^{\Oh(1)}$~time. The intuitive idea of this algorithm is to find the connected vertex pairs in~$Q_1$ via brute force and use an algorithm for~\textsc{Vanilla-BNSL} as a subroutine to find the arcs of~$A_Q$.

\begin{proposition} \label{Prop: compute optimal attic arc set}
Let~$I:=(N,\Fa,t,k)$ be an instance of~\textsc{$(\Pi_1+v)$-Moral BNSL}, and let~$S$, $Q_0$, and~$Q_1$ be disjoint subsets of~$N$. An optimal attic arc set for~$S$, $Q_0$, and~$Q_1$ can be computed in~$\lambda^{\Oh(\lambda)} \cdot |I|^{\Oh(1)}$~time, where~$\lambda:=|S \cup Q_0 \cup Q_1|$. 
\end{proposition}

\begin{proof}
Throughout this proof, let~$Q:=Q_0 \cup Q_1$ and~$N':=S \cup Q$. Consider~$Q_1$. An \emph{auxiliary graph}~$H$ is defined as an undirected graph with vertex set~$Q_1$, such that each connected component of~$H$ has size at most~2. Note that there are~$\binom{\lambda^2}{\lambda} \in \lambda^{\Oh(\lambda)}$ many auxiliary graphs, since~$|Q_1| \leq \lambda$.

Let~$H$ be a fixed auxiliary graph. For two vertices~$w_1 \in Q_1$ and~$w_2 \in Q_1$ we write~$w_1 \sim_H w_2$ if they belong to the same connected component of~$H$. In the following, we define a family~$\Fa^H$ of local scores for~$N'$. To this end, we introduce the term of \emph{feasible parent sets regarding~$H$}: First, let~$v \in Q_0$. A set~$P \subseteq N' \setminus \{v\}$ is called \emph{feasible for~$v$} if~$P \subseteq S$. Second, let~$v \in Q_1$. A set~$P \subseteq N' \setminus \{v\}$ is feasible for~$v$ if~$P \cap Q \subseteq \{w\}$ where~$w \sim_H v$. Finally, let~$v \in S$. A set~$P \subseteq N' \setminus \{v\}$ is \emph{feasible for~$v$}, if~$|P \cap Q| \leq 1$, or~$P \cap Q = \{w_1,w_2\}$ for some~$w_1, w_2 \in Q_1$ with~$w_1 \sim_H w_2$. We then define~$\Fa^H$ by
\begin{align*}
f^H_v(P) :=
\begin{cases}
f_v(P) & \text{if }P\text{ is feasible for }v\text{, or}\\
0 & \text{otherwise.}
\end{cases}
\end{align*}

Note that for every vertex~$v \in N'$, every potential parent set~$P \in \mathcal{P}_{\Fa^H}(v)$ is feasible for~$v$ by the definition of~$\Fa^H$.

\textit{Algorithm.} The algorithm to compute an optimal arc set for~$S$, $Q_0$, and~$Q_1$ can be described as follows: Iterate over all auxiliary graphs~$H$. For every choice of~$H$ compute an arc set~$A_H \subseteq N' \times N'$ that maximizes~$\sum_{v \in N'} f^H_v(P^{A_H}_v)$. Return an arc set~$A \in \{ A_H \mid H \text{ is an auxiliary graph}\}$ that maximizes~$\sum_{v \in N'}f_v(P^A_v)$.

\textit{Running time.} We first consider the running time of the algorithm. As mentioned above, we can iterate over all auxiliary graphs in~$\lambda^{\Oh(\lambda)}$~time. For every auxiliary graph, we compute the arc set~$A_H$. This can be done by solving~\textsc{Vanilla-BNSL} for the vertex set~$N'$ and local scores~$\Fa^H$. This can be done in~$2^\lambda \cdot |I|^{\Oh(1)}$~time~\cite{OM03,SM06}. Thus, the overall running time of the algorithm is~$\lambda^{\Oh(\lambda)} \cdot |I|^{\Oh(1)}$. 

\textit{Correctness.} It remains to show that the algorithm is correct. That is, the returned arc set~$A$ is an optimal attic arc set for~$S$, $Q_0$, and~$Q_1$. Note that~$A=A_H$ for some auxiliary graph~$H$. Therefore,~$A$ is a solution of an instance of~\textsc{Vanilla-BNSL} with vertex set~$N'$ and local scores~$\Fa^H$. By Proposition~\ref{Prop: Potential Parents}, we may assume that~$A$ is nice and therefore, for every~$v \in N'$ the parent set $P^A_v$ is feasible for~$v$ regarding~$H$. Consequently,~$f^H_v(P^A_v)=f_v(P^A_v)$ for every~$v \in N$.

We first show that~$A$ is an attic arc set for~$S$, $Q_0$, and~$Q_1$. That is, we show that Properties~$a)$ and~$b)$ from Definition~\ref{Def: ancestor extension} hold. Since~$A$ is a solution of a \textsc{Vanilla-BNSL} instance with vertex set~$N'$, the graph~$(N',A)$ is a DAG. Thus, Property~$a)$ from Definition~\ref{Def: ancestor extension} holds. We next check Property~$b)$. First, consider~$v \in Q_0$ and assume towards a contradiction~$v$ has a neighbor~$w \not \in S$ in~$\Mo(N',A)$. If~$(v,w) \in A$ or~$(w,v) \in A$, either~$v$ or~$w$ has a non-feasible parent set regarding~$H$. A contradiction. Otherwise, if~$\{v,w\}$ is a moral edge, then there exists a vertex~$u \in N'$ with~$\{v,w\} \in P^A_u$. Then,~$P^A_u$ is not feasible for~$u$ regarding~$H$. A contradiction. Second, consider~$v \in Q_1$. Then, due to the definition of feasible parent sets,~$v$ can only be adjacent to a vertex~$w \in Q \setminus \{v\}$ if~$v \sim_H w$. Since the connected components in~$H$ have size at most 2, $v$ has at most one neighbor outside S in~$\Mo(N',A)$. Therefore, Property~$b)$ from Definition~\ref{Def: ancestor extension} holds. Thus,~$A$ is an attic arc set.

We next show that~$A$ is optimal. That is, we show that~$\sum_{S \cup Q} f_v(P^A_v)$ is maximal among all attic arc sets for~$Q_0$,~$Q_1$, and~$S$. To this end, let~$A' \neq A$ be another attic arc set. Consider~$\Mo(N',A')$. Since every vertex in~$Q_1$ has at most one neighbor outside~$S$ in~$\Mo(N',A')$, the graph~$H':=\Mo(N',A')[Q_1]$ has connected components of size at most 2. Consequently,~$H'$ is an auxiliary graph. To show that~$\sum_{S \cup Q} f_v(P^{A'}_v) \leq \sum_{S \cup Q} f_v(P^A_v)$ we use the following claim.

\begin{claim} \label{Claim: Only Feasible Parent Sets}
For every~$n \in N'$, the parent set~$P^{A'}_v$ is feasible for~$v$ regarding the auxiliary graph~$H'$.
\end{claim}

\begin{proof}We consider the following cases.

\textbf{Case 1:} $v \in Q_0$\textbf{.} Then,~$v$ has no neighbors outside~$S$ in~$\Mo(N',A')$. Thus, $v$ has only incoming arcs from~$S$. Therefore,~$P^{A'}_v$ is feasible for~$v$.

\textbf{Case 2:} $v \in Q_1$\textbf{.} Then,~$v$ has at most one neighbor~$w$ outside~$S$ in~$\Mo(N',A')$. Observe that~$w \sim_{H'} v$ by the definition of~$H'$. Therefore,~$P^{A'}_v$ is feasible for~$v$.

\textbf{Case 3:} $v \in S$\textbf{.} Then, if~$|P^{A'}_v \cap Q| \leq 1$,~$P^{A'}_v$ is feasible for~$v$. Furthermore, if~$|P^{A'}_v \cap Q| \geq 3$, the vertices in~$P^{A'}_v \cap Q$ have degree at least 2 outside~$S$ in~$\Mo(N',A')$ contradicting the fact that~$A'$ is an attic arc set. Thus, it remains to consider the case where~$|P^{A'}_v \cap Q| = 2$. Let~$P^{A'}_v \cap Q = \{w_1,w_2\}$. Then,~$w_1$ and~$w_2$ are connected by a moral edge in~$\Mo(N',A')$ implying~$w_1 \sim_{H'} w_2$. Thus,~$P^{A'}_v$ is feasible for~$v$. $\hfill \Diamond$
\end{proof}
Since every~$v \in N'$ has a feasible parent set under~$A'$ regarding~$H'$, we have~$f^{H'}_v(P^{A'}_v) = f_v(P^{A'}_v)$. Since the score of~$A$ under~$\Fa^H$ is at least as big as the score of the best possible DAG under~$\Fa^{H'}$, we conclude
\begin{align*}
\sum_{S \cup Q} f_v(P^{A'}_v) = \sum_{S \cup Q} f^{H'}_v(P^{A'}_v) \leq \sum_{S \cup Q} f^{H}_v(P^{A}_v) = \sum_{S \cup Q} f_v(P^{A}_v).
\end{align*}$\hfill \Box$
\end{proof}

\subsection{Finding the Basement Arc Set} \label{Subsection: Find Basement} 
We now show that we can compute a basement with maximal score in polynomial time if~$S$,~$Q_0$, and~$Q_1$ are given. More precisely, we solve the following problem.
\begin{center}
	\begin{minipage}[c]{.9\linewidth}
          \textsc{Basement Learning}\\
          \textbf{Input}: A set of vertices~$N$, disjoint subsets~$S$,~$Q_0$,~$Q_1$ of~$N$, local scores~$\Fa=\{f_v \mid v \in N\}$, and an integer~$t$.\\
          \textbf{Question}: Is there a basement arc set~$A_R$ for~$S$,~$Q_0$, and~$Q_1$ with~$\sum_{v \in N \setminus (S \cup Q_0 \cup Q_1)} f_v(P_v^{A_R}) \geq t$?
	\end{minipage}
\end{center}

\begin{proposition} \label{Prop: Completion in Poly Time}
\textsc{Basement Learning} can be solved in $\Oh(n^3 \delta_\Fa)$~time.
\end{proposition}

\begin{proof}
We give a polynomial-time reduction to~\textsc{Maximum Weight Matching}\iflong . In~\textsc{Maximum Weight Matching} \fi one is given a graph~$G=(V,E)$, edge-weights~$\omega: E \rightarrow \mathds{N}$, and~$\ell \in \mathds{N}$ and the question is if there exists a set~$M \subseteq E$ of pairwise non-incident edges such that~$\sum_{e \in M} \omega(e) \geq \ell$. 

\textit{Construction:} Let~$I:=(N,S, Q_0, Q_1, \Fa, t)$ be an instance of~\textsc{Basement Learning}. Throughout this proof let~$Q:=Q_0 \cup Q_1$, and let~$R:= N \setminus (S \cup Q)$. We construct an equivalent instance~$(G, \omega, \ell)$ of~\textsc{Maximum Weight Matching}. We first define~$G:=(V,E)$ with~$V:=Q_0 \cup R \cup R'$, where~$R':=\{v' \mid v \in R\}$, and~$E:= X \cup Y \cup Z$, where
\begin{align*}
X &:= \{ \{v,w\} \mid v,w \in R, v \neq w\},\\
Y &:= \{ \{v,w\} \mid v \in R, w \in Q_0\}, \text{ and}\\
Z &:= \{ \{v,v'\} \mid v \in R\}.
\end{align*}
Next, we define edge-weights~$\omega : E \rightarrow \mathds{N}$: For~$e=\{v,v'\} \in Z$, we set

$$\omega(e) := \max_{S' \subseteq S} f_v(S').$$

 Furthermore, for~$e=\{v,w\} \in Y$ with~$v \in R$ and~$w \in Q_0$, we set
 
$$\omega(e) := \max_{S' \subseteq S} f_v(S' \cup \{w\}).$$

Finally, for~$e=\{v,w\} \in X$, we set~$\omega(e):= \max (\varphi(v,w), \varphi(w,v))$, where
\begin{align*}
\varphi(u_1,u_2) := &\max_{S' \subseteq S } f_{u_1}(S' \cup \{u_2\}) + \max_{S' \subseteq S } f_{u_2}(S').
\end{align*}
To complete the construction \iflong of~$(G, \omega, \ell)$\fi, we set~$\ell := t$.

\iflong \textit{Intuition:} Before we prove the correctness of the reduction we provide some intuition. \fi A maximum-weight matching~$M$ in~$G$ corresponds to the parent sets of vertices in~$R$ and therefore to arcs in a solution~$A_R$ of~$I$. More precisely, an edge~$\{v,v'\} \in Z$ with~$v \in R$ corresponds to a parent set of~$v$ that contains only vertices from~$S$. Moreover, an edge~$\{v,w\} \in Y$ with~$v \in R$ corresponds to a parent set of~$v$ that contains~$w \in Q_0$ and vertices from~$S$. Finally, an edge~$\{v,w\} \in X$ means that either~$v \in P^{A_R}_w$ or~$w \in P^{A_R}_v$. 
\iflong

\textit{Correctness:} 
We now prove the correctness of the reduction, that is, we show that~$I$ is a yes-instance of~\textsc{Basement Learning} if and only if~$(G,\omega,\ell)$ is a yes-instance of~\textsc{Maximum Weight Matching}.

$(\Rightarrow)$ Let~$A_R$ be a basement arc set of~$S$,~$Q_0$, and~$Q_1$ with~$\sum_{v \in R} f_v(P_v^{A_R}) \geq t$. We define a matching~$M$ with~$\sum_{e \in M} \omega(e) \geq t$. To this end, we describe which edges of~$X$,~$Y$, and~$Z$ we add to~$M$ by defining sets~$M_X$, $M_Y$, and~$M_Z$ and set~$M:= M_X \cup M_Y \cup M_Z$.

First, for every pair~$v,w \in R$ with~$v \in P^{A_R}_w$ or~$w \in P^{A_R}_v$, we add~$\{v,w\} \in X$ to~$M_X$. Second, for every pair~$v,w$ with~$v \in R$,~$w \in Q_0$, and~$w \in P^{A_R}_v$, we add~$\{v,w\}$ to~$M_Y$. Third, for every~$v \in R$ that is not incident with one of the edges in~$M_X \cup M_Y$, we add~$\{v,v'\}$ to~$M_Z$. Obviously, $M_X$, $M_Y$, and~$M_Z$ are pairwise disjoint. 

We first show that~$M$ is a matching by proving that there is no pair of distinct edges in~$M$ that share an endpoint. Consider the following cases.

\textbf{Case 1: $e_1, e_2 \in M_Z$.} Then, if~$e_1, e_2$ share one endpoint~$v \in R$ or~$v' \in R'$ it follows by the definition of~$M_Z$ that~$e_1=e_2=\{v,v'\}$ and, therefore, there are no distinct edges~$e_1, e_2 \in M_Z$ that share exactly one endpoint.

\textbf{Case 2: $e_1, e_2 \in M_X \cup M_Y$.} Then, assume towards a contradiction that~$e_1=\{u,v\}$ and~$e_2=\{v,w\}$ have a common endpoint~$v$. Now,~$\{u,v\} \in M_X \cup M_Y$ implies~$(u,v) \in A_R$ or~$(v,u) \in A_R$. Moreover~$\{v,w\} \in M_X \cup M_Y$ implies~$(w,v) \in A_R$ or~$(v,w) \in A_R$. Then,~$v \in R \cup Q_0$ is incident with two arcs in~$A_R \cap ((R\cup Q_0) \times R)$ which contradicts the fact that~$A_R$ is a basement arc set.

\textbf{Case 3: $e_1 \in M_X \cup M_Y, e_2 \in M_Z$.} Then,~$e_1$ and~$e_2$ can only have a common endpoint in~$R$ which is not possible by the definition of~$M_Z$.

We conclude by the above that~$M$ is a matching. It remains to show that~$\sum_{e \in M} \omega(e) \geq t$. Observe that every~$v \in R$ is incident with some edge in~$M$. Conversely, every edge in~$M_Y \cup M_Z$ has exactly one endpoint in~$R$, and every edge in~$M_X$ has both endpoints in~$R$. Given an edge~$e \in M_Y \cup M_Z$, we let~$\pi(e)$ denote its unique endpoint in~$R$. By the construction of~$M_X$ and the fact that~$A_R$ is a basement arc set we know that for every~$\{v,w\} \in M_X$ it holds that either~$(v,w) \in A_R$ or~$(w, v) \in A_R$. We let~$\pi_1(e)$ and~$\pi_2(e)$ denote the endpoints of~$e=\{v,w\}$ such that~$(\pi_2(e), \pi_1(e)) \in A_R$. Since every~$v \in R$ is incident with some edge in~$M$ and~$M$ is a matching, the following sets form a partition of~$R$.

\resizebox{0.95\linewidth}{!}{
  \begin{minipage}{\linewidth}
\begin{align*}
R_1 &: = \{ \pi_1(e) \mid e \in M_X \}, & R_2 &:=\{ \pi_2(e) \mid e \in M_X \},\\
R_3 &: = \{ \pi(e) \mid e \in M_Y \}, & R_4 &:=\{ \pi(e) \mid e \in M_Z \}.
\end{align*}
  \end{minipage}
}\\

Observe that by the definitions of~$M_X, M_Y$, and~$M_Z$ it holds that all~$v \in  R_2 \cup R_4$ have a parent set~$S'$ under~$A_R$, where~$S' \subseteq S$. Moreover, all~$\pi(e) \in R_3$ have parent set~$P^{A_R}_{\pi(e)} = S' \cup (e \setminus \{\pi(e)\})$ with~$S' \subseteq S$, and all~$\pi_1(e) \in R_1$ have parent sets~$P_{\pi_1(e)}^{A_R} = S' \cup \{\pi_2(e)\}$ with~$S' \subseteq S$. For the weight of~$M$ it then holds that
\begin{align*}
&\sum_{e \in M_X} \omega (e) + \sum_{e \in M_Y} \omega (e) + \sum_{e \in M_Z} \omega (e)\\
=&\sum_{e \in M_X} \max_{S' \subseteq S} f_{\pi_1(e)}(S' \cup \{\pi_2(e)\}) +\sum_{e \in M_X} \max_{S' \subseteq S} f_{\pi_2(e)}(S')\\
&+ \sum_{e \in M_Y} \max_{S' \subseteq S} f_{\pi(e)}(S' \cup (e \setminus \{\pi(e)\})) + \sum_{e \in M_Z} \max_{S' \subseteq S} f_{\pi(e)}(S')\\
\geq&\sum_{v \in R_1 \cup R_2 \cup R_3 \cup R_4} f_v(P^{A_R}_v) \geq t,
\end{align*}
and therefore~$\sum_{e \in M} \omega(e) \geq t$.

$(\Leftarrow)$ Conversely, let~$M \subseteq E$ be a matching of~$G$ with~$\sum_{e \in M} \omega(e) \geq t$. Note that in~$G$, every edge~$e \in E$ has at least one endpoint in~$R$ and consequently every~$e \in M$ has at least one endpoint in~$R$. Moreover, without loss of generality we can assume that every vertex of~$R$ is incident with an edge of~$M$: If a vertex~$v \in R$ is not incident with an edge of~$M$, we replace~$M$ by~$M':=M \cup \{\{v,v'\}\}$. Then,~$\sum_{e \in M'} \omega(e) \geq t+ \omega(\{v,v'\}) \geq t$ and~$M'$ is still a matching since~$\deg_G(v')=1$.

We define a set~$A_R \subseteq (S \cup Q_0 \cup R) \times R$ and show that~$\sum_{v \in R} f_v(P^{A_R}_v) \geq t$ and that~$A_R$ is a basement arc set. To this end, we define a parent set with vertices in~$S \cup Q_0 \cup R$ for every~$v \in R$. First, if~$v$ is incident with an edge~$\{v,v'\} \in M \cap Z$, we set~$P^{A_R}_v := \argmax_{S' \subseteq S} f_v(S')$. Second, if~$v$ is incident with an edge~$\{v,w\} \in M \cap Y$, then~$w \in Q_0$ and we set~$P^{A_R}_v := \{w\} \cup \argmax_{S' \subseteq S} f_v(S' \cup \{w\})$. Third, it remains to define the parent sets of vertices in~$R$ that are endpoints of some edge in~$M \cap X$. Let~$\{v,w\} \in M \cap X$, where~$\varphi(v,w) \geq \varphi(w,v)$. We then set~$P^{A_R}_v := \{w\} \cup \argmax_{S' \subseteq S} f_v(S' \cup \{w\})$ and we set~$P^{A_R}_w := \argmax_{S' \subseteq S} f_w(S')$.

We first show that~$A_R$ is a basement arc set. Obviously, $A_R$ does not contain self-loops and no~$v \in R$ has a parent in~$Q_1$. It remains to show that every vertex in~$Q_0 \cup R$ has at most one incident arc in~$A_R \cap ((R \cup Q_0) \times R)$. Let~$v \in Q_0 \cup R$. Assume towards a contradiction that~$v$ is incident with two distinct arcs in~$A_R \cap ((R \cup Q_0) \times R)$. Then, there exists a vertex~$w_1 \in Q_0 \cup R$ with~$(v,w_1) \in A_R$ or~$(w_1,v) \in A_R$. Moreover, there exists a vertex~$w_2 \in (Q_0 \cup R) \setminus \{w_1\}$ with~$(w_2,v) \in A_R$ or~$(v,w_2) \in A_R$. Then, by the definition of~$A_R$ we conclude~$\{v,w_1\},\{v,w_2\} \in M$ which contradicts the fact that no two edges in~$M$ share one endpoint. We conclude that~$A_R$ is a basement arc set.

It remains to show that~$\sum_{v \in R} f_v(P^{A_R}_v) \geq t$. To this end consider the following.

\begin{claim} \
\begin{enumerate}
\item[a)] If~$\{v,w\} \in M \cap (Y \cup Z)$ with~$v \in R$, then~$\omega(\{v,w\}) = f_v(P_v^{A_R})$.
\item[b)] If~$\{v,w\} \in M \cap X$, then~$\omega(\{v,w\}) = f_v(P_v^{A_R}) + f_w(P_w^{A_R})$.
\end{enumerate}
\label{Claim: Correctness Matching Reduction}
\end{claim}

\begin{proof}
$a)$ If~$\{v,w\} \in M \cap Z$, then~$w=v'$ and it follows by the definition of~$A_R$ that~$f_v(P_v^{A_R})= \max_{S' \subseteq S} f_v(S') = \omega(\{v,v'\})$. Otherwise, if~$\{v,w\} \in M \cap Y$, then~$w \in Q_0$ and analogously~$f_v(P_v^{A_R})= \max_{S' \subseteq S} f_v(S' \cup \{w\}) = \omega(\{v,w\})$.

$b)$ If~$\{v,w\} \in M \cap X$, then~$v,w \in R$. We only consider the case~$\varphi(v,w) \geq \varphi(w,v)$, since the other case is analogue. It then follows from the definition of~$A_R$, that
\begin{align*}
&f_v(P_v^{A_R}) + f_w(P_w^{A_R})\\
 =  &\max_{S' \subseteq S} f_v(S' \cup \{w\}) +  \max_{S' \subseteq S} f_w(S') \\
 =  &\max (\varphi(v,w), \varphi(w,v)) = \omega(\{v,w\}).
\end{align*}
$\hfill \Diamond$
\end{proof}

Now, let~$\tilde{R} \subseteq R$ be the set of vertices in~$R$ that are incident with an edge in~$M \cap X$. We conclude by Claim~\ref{Claim: Correctness Matching Reduction} and the assumption that every~$v \in R$ is incident with an edge in~$M$ that
\begin{align*}
\sum_{v \in R} f_v(P^{A_R}_v) &= \sum_{v \in R\setminus \tilde{R}} f_v(P^{A_R}_v) + \sum_{v \in \tilde{R}} f_v(P^{A_R}_v)\\
&= \sum_{e \in M \cap (Y \cup Z)} \omega(e) + \sum_{e \in M \cap X} \omega(e)\\
&= \sum_{e \in M} \omega(e) \geq  t,
\end{align*}
which completes the correctness proof.

\textit{Running Time.} The constructed instance of~\textsc{Maximum Weight Matching} contains~$\Oh(n)$ vertices and~$\Oh(n^2)$ edges. For each edge~$e$, the edge weight~$\omega(e)$ can be computed in~$\Oh(n\delta_\Fa)$ time. Hence, we can compute the described instance of \textsc{Maximum Weight Matching} from an instance of~\textsc{Basement Learning} in~$\Oh(n^3 \delta_\Fa)$ time. Together with the fact that \textsc{Maximum Weight Matching} can be solved in~$\Oh(\sqrt{|V|}\cdot |E|)$ time~\cite{MV80}, we conclude that~\textsc{Basement Learning} can be solved in~$\Oh(n^3\delta_\Fa)$~time.
\fi
 $\hfill \Box$
\end{proof}

\subsection{An XP-time Algorithm for \textsc{$\mathbf{(\Pi_1+v)}$-Moral BNSL}} \label{Subsection: Bounded Diss-No Algo}

We now combine the previous results of this section to obtain an XP-time algorithm for~\textsc{$(\Pi_1+v)$-Moral BNSL}. Recall that the intuitive idea of the algorithm is to find the dissociation set~$S$ and the ancestors~$Q= Q_0 \cup Q_1$ of~$S$ via brute force. Then, for every choice of~$S$,~$Q_0$, and~$Q_1$ we find an attic arc set and a basement arc set and combine these arc sets to a solution of \textsc{$(\Pi_1+v)$-Moral BNSL}.

\begin{theorem} \label{Theorem: XP Algo}
\textsc{$(\Pi_1+v)$-Moral BNSL} can be solved in~$n^{3k} \cdot k^{\Oh(k)} \cdot |I|^{\Oh(1)}$~time.
\end{theorem}

\iflong
\begin{proof}
\textit{Algorithm.} 
Let~$I=(N,\Fa,t,k)$ be an instance of~\textsc{$(\Pi_1+v)$-Moral BNSL}. The following algorithm decides whether~$I$ is a yes-instance or a no-instance: First, iterate over all possible choices of~$S$,~$Q_0$, and~$Q_1$ where~$|S| \leq k$ and~$|Q_0 \cup Q_1| \leq 2k$. For each such choice do the following:
\begin{enumerate}
\item[•] Compute an optimal attic arc set~$A_Q$ using the algorithm behind Proposition~\ref{Prop: compute optimal attic arc set}.

\item[•] Let~$t':= t- \sum_{v \in S\cup Q} f_v(P_v^{A_Q})$ and check if~$(N,S,Q_0,Q_1, \Fa, t')$ is a yes-instance of~\textsc{Basement Learning}. If this is the case, return \emph{yes}.
\end{enumerate}
If for none of the choices of~$S$,~$Q_0$, and~$Q_1$ the answer~\emph{yes} was returned, then return~\emph{no}.

\textit{Running Time.} First, we discuss the running time of the algorithm: Since~$|S| \leq k$ and~$|Q_0 \cup Q_1| \leq 2k$, there are at most~$\binom{n}{k} \cdot \binom{n-k}{2k} \in \Oh(n^{3k})$~choices for~$S$ and~$Q:=Q_0 \cup Q_1$. For each such choice we can compute all~$4^k$ possible partitions of~$Q$ into two sets. Hence, we can iterate over all possible choices for~$S$,~$Q_0$ and~$Q_1$ in $\Oh(n^{3k} \cdot 4^k)$~time. Afterwards, for each such choice we apply the algorithm behind Proposition~\ref{Prop: compute optimal attic arc set} in~$k^{\Oh(k)} \cdot |I|^{\Oh(1)}$~time, and the algorithm behind Proposition~\ref{Prop: Completion in Poly Time} in~$\Oh(n^3 \delta_{\Fa})$~time. This gives an overall running time of~$n^{3k} \cdot k^{\Oh(k)} \cdot |I|^{\Oh(1)}$ as~claimed.

\textit{Correctness.} Second, we show the correctness of the algorithm by proving that the algorithm returns yes if and only if~$I$ is a yes-instance of~\textsc{$(\Pi_1+v)$-Moral BNSL}. 

$(\Rightarrow)$ Let the algorithm return yes for~$I$. Then, there exist disjoint sets~$S$,~$Q_0$, and~$Q_1$ with~$|S| \leq k$, an attic arc set~$A_Q$, and a basement arc set~$A_R$ such that~$\sum_{v \in S \cup Q} f_v(P_v^{A_Q}) + \sum_{v \in R} f_v(P_v^{A_R}) \geq t$. Then, by Lemma~\ref{Lemma: segmentation -> DAG}, the graph~$(N,A_Q \cup A_R)$ is a DAG whose moralized graph has dissociation set~$S$ and by Lemma~\ref{Lemma: Decompose Score} its score is the sum of the local scores for~$A_R$ and~$A_Q$. Hence,~$I$ is a yes-instance.

$(\Leftarrow)$ Conversely, let~$(N,\Fa,t,k)$ be a yes-instance. Then, there is an~$(N,\Fa,t)$-valid DAG~$D=(N,A)$ whose moralized graph has a dissociation set~$S$ of size at most~$k$. Then, by Lemma~\ref{Lemma: DAG -> segmentation}, there exist disjoint sets~$Q_0$,~$Q_1$, an attic arc set~$A_Q$ and a basement arc set~$A_R$ such that~$A_R \cup A_Q= A$. Furthermore,~$|Q_0 \cup Q_1| \leq 2k$. Since the algorithm iterates over all possible choices for~$S$,~$Q_0$, and~$Q_1$ with~$|S| \leq k$ and~$|Q_0 \cup Q_1| \leq 2k$, it considers~$S$,~$Q_0$, and~$Q_1$ at some point. Since~$A$ is~$(N,\Fa,t)$-valid, the arc set~$A_R$ satisfies~$\sum_{v \in R} f_v(P^{A_R}_v) \geq t- \sum_{v \in S\cup Q} f_v(P_v^{A_Q})$. Hence, the algorithm returns yes. $\hfill \Box$
\end{proof}
\fi

The running time stated in Theorem~\ref{Theorem: XP Algo} contains a factor of~$k^{\Oh(k)}$. Let us remark that the constant in the exponent hidden by the~$\Oh$-notation is not too high: The constant relies on the running time of the algorithm behind Proposition~\ref{Prop: compute optimal attic arc set} where we iterate over the possible auxiliary graphs. Since~$|Q_1| \leq 2k$, the number of iterations is~$\Oh(\binom{(2k)^2}{2k})$. Thus, the hidden constant is 4. While it seems possible that this can be improved, it would be more interesting to determine whether the factor of~$k^{\Oh(k)}$ can be replaced by~$2^{\Oh(k)}$.

\iflong
\subsection{\textsc{$\mathbf{(\Pi_1+v)}$-Moral BNSL} Parameterized by~$\mathbf{k+t}$} \label{Subsection: DissNoBNSL W[1]-h for k+t}
\fi
In this section, we show another hardness result for~\textsc{$(\Pi_1+v)$-Moral BNSL}. Due to Corollary~\ref{Cor: Bounded-VC W[2]-h}, the problem is presumably not FPT for parameter~$k$ even if the maximum parent set size is~1. Observe that this implies that even a parameterization by~$k+p$, where~$p$ denotes the maximum parent set size, is unlikely to lead to an FPT algorithm. In the next theorem, we prove W[1]-hardness for~$k+t$, even if~$p=3$ and~$\delta_\Fa=2$ implying that not even for the sum~$k+t+\delta_\Fa+p$ there is an FPT algorithm unless~$\text{W[1]}=\text{FPT}$. The proof is closely related to the W[1]-hardness proof for \textsc{$(\Pi_0+v)$-Moral BNSL}~\cite{KP15}. We provide it here for sake of completeness.

\begin{theorem} \label{Theorem: Bounded-DissNo W[1]-h}
\textsc{$(\Pi_1+v)$-Moral BNSL} is W[1]-hard for~$k+t$, even when the superstructure~$S_{\vec{\Fa}}$ is a DAG, the maximum parent set size is $3$,~$\delta_\Fa=2$, and every local score is either 1 or 0.
\end{theorem}

\iflong
\begin{proof} 
We give a parameterized reduction from~\textsc{Clique}. In~\textsc{Clique} we are given an undirected graph~$G=(V,E)$ and an integer~$\ell$ and the question is if there exists a subset~$K \subseteq V$ of size~$\ell$ such that the vertices in~$K$ are pairwise adjacent in~$G$. \textsc{Clique} is W[1]-hard if parameterized by~$\ell$~\cite{DRGM95}. Let~$(G=(V,E),\ell)$ be an instance of \textsc{Clique}. We describe how to construct an equivalent instance of~\textsc{$(\Pi_1+v)$-Moral BNSL} where~$k+t \in \Oh(\ell^2)$.

\textit{Construction.} We first define the vertex set~$N$ by~$N:= V \cup (E \times \{1\}) \cup (E \times \{2\})$. We write~$e_i:=(e,i)$ for the elements in~$E \times \{i\}$, $i \in \{1,2\}$. Next, we define the local scores~$\Fa$. For every edge~$e=\{u,v\} \in E$ we set~$f_{e_1}(\{e_2,u,v\}):=1$. All other local scores are set to~$0$. Finally, we set~$t:=\binom{\ell}{2}$ and~$k:=\ell$. Note that~$k+t \in \Oh(\ell^2)$ and that the maximum parent set size is 3. Moreover, in the superstructure~$S_{\vec{\Fa}}$, the vertices in~$E \times \{1\}$ are the only vertices that have incoming arcs, and all vertices in~$E \times \{1\}$ are sinks. Hence, $S_{\vec{\Fa}}$ is a DAG.

\textit{Correctness.} We next show that there is a clique of size~$k$ in~$G$ if and only if there exists an~$(N,\Fa,t)$-valid arc set~$A$ such that~$\Mo(N,A)$ has a dissociation set of size at most~$k$.

$(\Rightarrow)$ Let~$K$ be a clique of size~$\ell$ in~$G$. We then define~$A:=\{(e_2,e_1),(u,e_1),(v,e_1) \mid e=\{u,v\} \text{ with }u,v \in K\}$. We prove that~$A$ is~$(N,\Fa,t)$-valid and~$\Mo(D)$ for~$D:=(N,A)$ has a dissociation set of size at most~$k$. 

From the fact that~$S_{\vec{\Fa}}$ is a DAG, we conclude that~$D$ is a DAG. Moreover, since the vertices of~$K$ are pairwise adjacent in~$G$, it follows by the construction of~$\Fa$ that~$f_{e_1}(P^A_{e_1})=1$ for all~$e \in E_G(K)$. Since~$|K|=\ell$ we conclude~$\sum_{v \in N} f_v(P^A_v) = \binom{\ell}{2} = t$. Hence, $A$ is~$(N,\Fa,t)$-valid.

Finally, we show that~$K$ is a dissociation set of size~$k$ in~$\Mo(D)$. Observe that the vertices~$e_1\in N$ with~$e \in E_G(K)$ are the only vertices that have a non-empty parent set in~$D$. Hence, the moralized graph is~$\Mo(D)=(N,E_\Mo)$, where
\begin{align*}
E_{\Mo} = \bigcup_{e:=\{u,v\} \in E_G(K)} \{\{w_1,w_2\} \mid w_1, w_2 \in \{u,v,e_1,e_2\}\}.
\end{align*}
It follows that~$\Mo(D)-K=(N\setminus K, E_\Mo')$ with
\begin{align*}
E_{\Mo}' = \bigcup_{e \in E_G(K)} \{\{e_1,e_2\}\}.
\end{align*}
Therefore, the maximum degree in~$\Mo(D)-K$ is at most 1.

$(\Leftarrow)$ Conversely, let~$A$ be an~$(N,\Fa,t)$-valid arc set such that the moralized graph~$\Mo(D)$, where~$D:=(N,A)$, has a dissociation set~$S$ of size at most~$k$. We prove that there is a clique of size~$k$ in~$G$.

Since~$A$ is~$(N,\Fa,t)$-valid, we know that~$\sum_{v \in N} f_v(P^A_v) \geq \binom{k}{2}$. Together with the definition of~$\Fa$ it follows that there are at least~$\binom{k}{2}$ vertices~$e_1 \in N$ such that~$P_{e_1}^A=\{e_2,u,v\}$ with~$\{u,v\}=e$. Let~$X \subseteq N$ be the set of these vertices~$e_1$. Note that~$|X|\geq\binom{k}{2}$ and recall that for every~$e_1 \in X$ it holds that~$e$ is an edge of~$G$.

For every~$e_1 \in X$ with~$e=\{u,v\}$ the set~$\{e_1, e_2, u,v\}$ is a clique of size 4 in~$\Mo(D)$. Since~$S$ is a dissociation set in~$\Mo(D)$, at least two vertices from each set~$\{e_1, e_2, u,v\}$ belong to~$S$. We may assume that these vertices are~$u$ and~$v$ by a simple exchange argument.

Now, consider~$S \cap V$. Since for every~$e_1 \in X$ with~$e=\{u,v\}$ the vertices~$u$ and~$v$ belong to~$S$, we conclude that~$|E_G(S \cap V)| \geq \binom{k}{2}$. Together with the fact that~$|S| \leq k$, it follows that~$S \cap V$ is a set of~$k$ vertices with at least~$\binom{k}{2}$ edges between them. Therefore, $S \cap V$ is a clique of size~$k$ in~$G$. $\hfill \Box$
\end{proof}
\fi

\section{Constrained BNSL for Related Graph Classes}

In this section we outline the limits of learning Bayesian networks under a sparsity constraints that are related to a bounded vertex cover number and a bounded dissociation number. Recall that a bound on the vertex cover number or the dissociation number automatically implies a bound on the treewidth, and that for efficient inference it is desirable to have a small treewidth in the moralized graph~\cite{D09}. Thus, it is well motivated to study constrained BNSL problems regarding graph classes that give a bound on the treewidth that lies between the treewidth and the dissociation number. 

 In terms of graph-classes, the bound on the vertex cover number is formalized as the graph class~$\Pi_0 + kv$ and the bound on the dissociation number is formalized as the graph class~$\Pi_1 +kv$. Recall that~$\Pi_0$ is the class of edgeless graphs. Equivalently,~$\Pi_0$ is the class of graphs with maximum degree 0, or the graphs with maximum connected component size 1. Analogously,~$\Pi_1$ is the class of graphs with maximum degree 1, or the class of graphs with maximum connected component size 2. In this secction, we consider two superclasses of~$\Pi_1$ and show that~XP-time algorithms for constrained BNSL problems regarding these superclasses are presumably not possible.

Let~$\Pi_2$ be the class of graphs that have maximum degree 2, and let~$\Pi_{c}^{\text{COC}}$ be the class of graphs where each connected component has size at most~$c$ for a fixed integer~$c \geq 3$. These graph classes are superclasses of~$\Pi_1$, that is~$\Pi_1 \subseteq \Pi_2$ and~$\Pi_1 \subseteq \Pi_c^\text{COC}$. Consequently, if a graph~$G$ belongs to the graph class~$\Pi_1+kv$ for some~$k \in \mathds{N}_0$, then there exist~$k'\leq k$ and~$k'' \leq k$ such that~$G \in \Pi_2 + k'v$ and~$G \in \Pi_{c}^{\text{COC}}+k''v$. Moreover, observe that the treewidth of~$G$ is not bigger than~$\min(k',k'')+\Oh(1)$. 

With the next two theorems we show that there is little hope that~\textsc{$(\Pi + v)$-Skeleton BNSL} or~\textsc{$(\Pi + v)$- Moral BNSL} with~$\Pi \in \{\Pi_2,\Pi_{3}^{\text{COC}}\}$ has an XP-time algorithm when parameterized by~$k$. To prove the first result we use a reduction from \textsc{Hamiltonian Path}. This construction was already used to show that \textsc{BNSL} is NP-hard if one adds the restriction that the resulting network must be a directed path~\cite{M01}. In the following we show that it also works for~\textsc{$(\Pi_2 + v)$-Skeleton BNSL} and \textsc{$(\Pi_2 + v)$-Moral BNSL}.

\begin{theorem} \label{Theorem: Bounded-DelToDeg2 NPh}
\textsc{$(\Pi_2 + v)$-Skeleton BNSL} and \textsc{$(\Pi_2 + v)$-Moral BNSL} are NP-hard even if~$k=0$ and the maximum parent set size is 1.
\end{theorem}

\iflong
\begin{proof} 
We give a polynomial-time reduction from the NP-hard~\textsc{Hamiltonian Path} problem to~\textsc{$(\Pi_2 + v)$-Skeleton BNSL}. Afterwards we show that the reduction is also correct for~\textsc{$(\Pi_2 + v)$-Moral BNSL}. In \textsc{Hamiltonian Path} one is given an undirected graph~$G$ and the question is whether there exists a Hamiltonian path, that is, a path which contains every vertex of~$G$ exactly once.

\textit{Construction.} Let~$G=(V,E)$ be an instance of~\textsc{Hamiltonian Path} with $n$~vertices. We describe how to construct an equivalent instance of~\textsc{$(\Pi_2+v)$-BNSL} where~$k=0$. We first set~$N:=V$. Next, for every~$v \in N$ we set~$f_v(\{w\})=1$ if~$w \in N_G(v)$ and~$f_v(P)=0$ for every other~$P \subseteq N \setminus \{v\}$. Finally, we set~$t:= n-1$ and~$k:=0$.

\textit{Correctness.} We next show that~$G$ is a yes-instance of~\textsc{Hamiltonian Path} if and only if~$(N,\Fa,t,0)$ is a yes-instance of~\textsc{$(\Pi_2+v)$-Skeleton BNSL}.

 $(\Rightarrow)$ Let~$P=(v_1, v_2, \dots, v_n)$ be a Hamiltonian path in~$G$. We set~$A:=\{(v_i,v_{i+1}) \mid i \in \{1, \dots, n-1 \}\}$ and show that~$A$ is~$(N, \Fa, t)$-valid and that~$\Sk(N,A)$ has maximum degree 2.

Since~$P$ is a Hamiltonian path, no vertex appears twice on~$P$. Hence,~$(N,A)$ does not contain directed cycles. Moreover, it holds that~$v_i \in N_G(v_{i-1})$ for every~$i \in \{2, \dots, n\}$ and therefore~$\sum_{v \in N} f_v(P^A_v) = n-1 = t$. Hence,~$A$ is~$(N, \Fa, t)$-valid. Moreover, observe that~$\Sk(N,A)= (N, \{\{v_i,v_{i+1}\} \mid i \in \{1, \dots, n-1\})$ and this,~$\Sk(N,A)$ has maximum degree 2.

$(\Leftarrow)$ Conversely, let~$A$ be an~$(N, \Fa, t)$-valid arc set such that~$\Sk(D)$ has maximum degree at most 2, where~$D:=(N,A)$. 
Since~$t=n-1$ and every local score is either 1 or 0 we conclude that~$f_v(P^A_v)=1$ for at least~$n-1$ vertices. Then, there are at least~$n-1$ arcs in~$A$, and thus there are at least~$n-1$~edges in~$\Sk(D)$.
Furthermore, we may assume that in~$D$ no vertex~$v$ has a non-empty parent set with score~$f_v(P^A_v)=0$ since otherwise we may replace it with~$\emptyset$. This implies that no vertex has more than one parent in~$D$. Consequently,~$\Sk(D)$ is acyclic.

Since~$\Sk(D)$ is acyclic, the maximum degree is 2, and there are at least~$n-1$ edges, there is a Hamiltonian path~$P:=(v_1, \dots, v_n)$ in~$\Sk(D)$. We show that~$P$ is a Hamiltonian path in~$G$. 
Let~$v_i$, $v_{i+1}$ be two consecutive vertices on~$P$. Then, either~$(v_i,v_{i+1}) \in A$ or~$(v_{i+1},v_{i}) \in A$. By the construction of~$\Fa$ we have~$v_i \in N_G(v_{i+1})$ and thus,~$\{v_i,v_{i+1}\} \in E$. Therefore,~$P$ is a Hamiltonian path in~$G$.

\textit{Moralized Graph.} We next argue why the construction described above is also a correct reduction from~\textsc{Hamiltonian Path} to \textsc{$(\Pi_2+v)$-Moral BNSL}.

For the forward direction, let~$P$ be a Hamiltonian path in~$G$ and let the arc set~$A$ be defined as above. Since every vertex has at most one incoming arc from~$A$, the moralized graph~$\Mo(N,A)$ has no moral edges and therefore~$\Mo(N,A)$ and~$\Sk(N,A)$ have the same set of edges. Thus,~$\Mo(N,A)$ has maximum degree 2.

For the backwards direction, let~$A$ be an~$(N, \Fa, t)$-valid arc set such that the moralized graph $\Mo(N,A)$ has maximum degree at most 2. Since the edge set of the skeleton of~$(N,A)$ is a subset of the edge set of~$\Mo(N,A)$, we conclude that~$\Sk(N,A)$ has maximum degree at most 2. Then, by the above argumentation, there exists a Hamiltonian path in~$G$. 
 $\hfill \Box$
\end{proof}

\begin{theorem} \label{Theorem: Bounded-3coc NP-h}
Let~$c \geq 3$. Then, \textsc{$(\Pi_{c}^{\text{COC}} + v)$-Skeleton BNSL} and \textsc{$(\Pi_{c}^{\text{COC}} + v)$-Moral BNSL} are NP-hard even if~$k=0$.
\end{theorem}

\begin{proof}
We give a polynomial-time reduction from the NP-hard problem~\textsc{$c$-Clique Cover}~\cite{KH83} to~\textsc{$(\Pi_{c}^{\text{COC}} + v)$-Skeleton BNSL}. Afterwards, we show that the reduction is also correct for~\textsc{$(\Pi_{c}^{\text{COC}} + v)$-Moral BNSL}. In~\textsc{$c$-Clique Cover} one is given an undirected graph~$G=(V,E)$ and the question is whether there exists a packing~$\mathcal{P}$ of vertex-disjoint cliques of size~$c$ such that every vertex of~$G$ belongs to one clique of the packing. We represent the packing of cliques~$\mathcal{P}:= \{ K^i_c \mid i \in \{1, \dots, \frac{|V|}{c} \} \}$ as a partition of~$V$, where every~$K^i_c$ is a clique of size~$c$. 

\textit{Construction.} Let~$G=(V,E)$ be an instance of~\textsc{$c$-Clique Cover} with~$n$ vertices. We describe how to construct an equivalent instance of~\textsc{$(\Pi^\text{COC}_c+v)$-Skeleton BNSL} where~$k=0$. We first set~$N:=V$. Next, for every~$v \in N$ we set~$f_v(P)=1$ if~$G[P]$ is a~$K_{c-1}$ and~$P \subseteq N_G(v)$. Otherwise, we set~$f_v(P)=0$. Note that~$\mathcal{F}$ can be computed in polynomial time since~$c$ is a constant. Finally, we set~$t:=\frac{n}{c}$ and~$k:=0$.

\textit{Correctness.} We next show that~$G$ is a yes-instance of~\textsc{$c$-Clique Cover} if and only if~$(N,\Fa,t,0)$ is a yes-instance of~\textsc{$(\Pi^{\text{COC}}_c+v)$-Skeleton BNSL}.

 $(\Rightarrow)$ Let~$\mathcal{P}=\{ K^i_c \mid i \in \{1, \dots, \frac{n}{c} \} \}$ be a packing of vertex-disjoint cliques of size~$c$ that cover~$G$. For each~$i$, let~$v_i$ be one vertex of~$K^i_c$. We set~$A:= \bigcup_{i \in \{1, \dots, \frac{n}{c}\}} \{(u,v_i) \mid u \in K^i_c \setminus \{v_i\}\}$. Then,~$D:=(N,A)$ is a union of disjoint stars, where the arcs of each star point to the center~$v_i$. Therefore, $D$ is a DAG and every connected component of~$\Sk(D)$ has order~$c$. Therefore,~$\Sk(D) \in \Pi^\text{COC}_c$. Finally, observe that~$f_{v_i}(P^A_{v_i})=1$ for each vertex~$v_i$ by the definition of~$\Fa$. Consequently,~$\sum_{v \in N} f_v(P^A_v) \geq \frac{n}{c}=t$.

$(\Leftarrow)$ Conversely, let~$A$ be an~$(N,\Fa,t)$-valid arc set such that every connected component of~$\Sk(N,A)$ has order at most~$c$. Since~$t=\frac{n}{c}$ and the local scores are either 0 or 1, there are pairwise distinct vertices~$v_1, \dots, v_\frac{n}{c}$ with~$f_v(P^A_{v_i})=1$ for every~$i \in \{1, \dots, \frac{n}{c}\}$. We define~$\mathcal{P}:=\{ K^i:= \{v_i\} \cup P^A_{v_i} \mid i \in \{1, \dots, \frac{n}{c}\}$ and show that~$\mathcal{P}$ is a packing of vertex disjoint size~$c$ cliques that cover~$G$.

By the definition of~$\Fa$, each~$P^A_{v_i}$ is a clique of size~$c-1$ that is completely contained in~$N_G(v_i)$. Thus, every~$K^i \in \mathcal{P}$ is a size-$c$ clique in~$G$. Next, assume towards a contradiction that there are distinct indices~$i$ and~$j$ with~$K^i \cap K^j \neq \emptyset$. Then, since~$v_i \neq v_j$ it follows that~$K^i \cup K^j$ is a connected component of size bigger than~$c$ in~$\Mo(N,A)$. This contradicts the choice of~$A$.

\textit{Moralized Graph.} We next argue why the construction described above is also a correct reduction from~\textsc{$c$-Clique Cover} to~\textsc{$(\Pi_3^\text{COC}$+v)-Moral BNSL}.

$(\Rightarrow)$ Let~$\mathcal{P}$ be a packing of vertex disjoint size-$c$ cliques that cover~$G$, and let~$A$ be defined as above. Then,~$(N,A)$ is a disjoint union of stars that point to the center of~$v_i$. Thus,~$\Mo(N,A)$ is a disjoint union of cliques of size~$c$. Consequently~$\Mo(N,A) \in \Pi^\text{COC}_c$.

$(\Leftarrow)$ Let~$A$ be an~$(N,\Fa,t)$-valid arc set such that every connected component of~$\Mo(N,A)$ has order at most~$c$. Since the edge set of the skeleton of~$(N,A)$ is a subset of the edge set of the moralized graph, we conclude that~$\Sk(N,A) \in \Pi^\text{COC}_c$. Then, by the above argumentation,~$G$ is a yes-instance of~\textsc{$c$-Clique Cover}.
  $\hfill \Box$
\end{proof}
\fi

Note that, given a graph class~$\Pi$, we have~$\Pi = \Pi+0v = \Pi+0e$. Thus, Theorems~\ref{Theorem: Bounded-DelToDeg2 NPh} and~\ref{Theorem: Bounded-3coc NP-h} imply the following.

\begin{corollary}
Let~$\Pi \in \{\Pi_2\} \cup \{\Pi^{\text{COC}}_c \mid c \geq 3\}$. Then, \textsc{$(\Pi + e)$-Skeleton BNSL} and~\textsc{$(\Pi + e)$-Moral BNSL} are NP-hard even if~$k=0$.
\end{corollary}

\section{BNSL with Bounded Number of Edges}

In this section we study BNSL, where we aim to learn a network such that the skeleton or the moralized graph have a bounded number of edges. Formally, we study~\textsc{$(\Pi_0+e)$-Skeleton BNSL} and~\textsc{$(\Pi_0+e)$-Moral BNSL}, where~$\Pi_0$ is the class of edgeless graphs. Clearly,~$\Pi_0$ is monotone.

We first consider~\textsc{$(\Pi_0+e)$-Skeleton BNSL} in Subsection~\ref{Subsection: Bounded Edges Skeleton}. We show that it is fixed-parameter tractable when parameterized by~$k$. Afterwards we consider~\textsc{$(\Pi_0+e)$-Moral BNSL} in Subsection~\ref{Subsection: Bounded Edges Moral}. We observe that it has an XP-time algorithm when parameterized by~$k$ and it is~W[1]-hard for parameterization by~$k+t$. Thus, putting the constraint of a bounded number of edges on the moralized graph makes the learning problem harder than putting a similar constraint on the skeleton.

\subsection{\textsc{$\mathbf{(\Pi_0+e)}$-Skeleton BNSL}} \label{Subsection: Bounded Edges Skeleton}		
In this subsection we consider a version of \textsc{Bayesian Network Structure Learning} where we want to learn a Bayesian network with a bounded number of arcs or---equivalently---a bounded umber of edges in the skeleton. Formally, this is the constrained BNSL problem~\textsc{$(\Pi_0+e)$-Skeleton BNSL}, where~$\Pi_0$ is the class of edgeless graphs.

\paragraph{A Polynomial-Time Algorithm when the Superstructure is Acyclic.} We first show that~\textsc{$(\Pi_0+e)$-Skeleton BNSL} becomes polynomial-time solvable if the superstructure is a DAG. The algorithm uses dynamic programming over a topological ordering of~$S_{ \vec{\Fa}}$. That is, an ordering~$(v_1, \dots, v_n)$ of the vertices of~$N$ such that for every arc~$(v_i,v_j)$ of~$S_{\vec{\Fa}}$ it holds that~$i<j$. 

\begin{proposition} \label{Prop: BA-BNSL poly if Sf DAG}
\textsc{$(\Pi_0+e)$-Skeleton BNSL} can be solved in~$\Oh(\delta_\Fa \cdot k \cdot n)$ time if the superstructure is~a~DAG.
\end{proposition}

\iflong
\begin{proof} Let~$N:=\{1, \dots, n\}$, and let~$(N,\Fa,t,k)$ be an instance of~\textsc{$(\Pi_0+e)$-Skeleton BNSL} such that~$S_{\vec{\Fa}}$ is a DAG. Without loss of generality, let~$(n,n-1, \dots,2,1)$ be a topological ordering of~$S_\Fa$. Hence, for every arc~$(a,b)$ of~$S_{\vec{\Fa}}$ it holds that~$a>b$.

The dynamic programming table~$T$ has entries of the type~$T[i,j]$ for all~$i \in \{0, 1, \dots, n\}$ and~$j \in \{0, 1, \dots, k\}$. Each entry stores the maximum sum of local scores of the vertices~$(i, \dots, 1)$ of the topological ordering that can be obtained by an arc set~$A$ of size at most~$j$. For~$i=0$, we set~$T[0,j]=0$ for all~$j \in \{0, \dots, k\}$. The recurrence to compute an entry for~$i>0$ is
\begin{align*}
T[i,j] = \max_{P \in P_{\Fa}(i), |P| \leq j} (f_i(P) + T[i-1, j- |P|])\text{,}
\end{align*}
and the result can then be computed by checking if~$T[n,k] \geq t$. The corresponding network can be found by traceback. The correctness proof is straightforward and thus omitted. The size of~$T$ is~$\Oh(n \cdot k)$ and each entry~$T[i,j]$ can be computed in~$\Oh(\delta_\Fa)$ time by iterating over the at most~$\delta_\Fa$ triples~$(f_i(P), |P|, P)$ in~$\Fa$ for the vertex~$i$. Therefore, \textsc{$(\Pi_0+e)$-Skeleton BNSL} can be solved in~$\Oh(\delta_\Fa \cdot k \cdot n)$ time if~$S_{\vec{\Fa}}$ is a~DAG. $\hfill \Box$
\end{proof}
\fi

 \paragraph{A Randomized FPT Algorithm.} 
 The dynamic programming algorithm behind Proposition~\ref{Prop: BA-BNSL poly if Sf DAG} can be adapted to obtain an FPT algorithm for~\textsc{$(\Pi_0+e)$-Skeleton BNSL} when parameterized by the number of arcs~$k$. 
The algorithm is based on color coding~\cite{AYZ95}: In a Bayesian network with at most~$k$ arcs, there are at most~$2k$ vertices which are endpoints of such arcs. The idea of color coding is to randomly color the vertices of~$N$ with~$2k$ colors and find a solution~$A$ where all vertices that are incident with arcs of~$A$ are colored with pairwise distinct colors.
\iflong 

To describe the color coding algorithm, we introduce some notation. Let~$N$ be a set of vertices. A function~$\chi: N \rightarrow \{1,\dots,2k\}$ is called a~\emph{coloring (of~$N$ with~$2k$ colors)}. Given a color~$c \in \{1,\dots,2k\}$, we call~$\chi^{-1}(c):=\{v \in N \mid \chi(v)=c\}$ the~\emph{color class of~$c$}. For a subset~$N' \subseteq N$, we let~$\chi(N'):=\{\chi(v) \mid v \in N'\}$, and for a subset~$C \subseteq \{1,\dots,2k\}$ we let~$\chi^{-1}(C):= \bigcup_{c \in C} \chi^{-1}(c)$. The following definition is important for our algorithm.

\begin{definition}\label{Def: Color-Loyal}
Let~$N$ be a set of vertices and let~$\chi: N \rightarrow \{1, \dots, 2k\}$ be a coloring of~$N$. An arc set~$A \subseteq N \times N$ is called~\emph{color-loyal for~$\chi$} if for every color class~$\chi^{-1}(c)$ it holds that
\begin{enumerate}
\item[a)] there is no~$(v,w) \in A$ with~$v,w \in \chi^{-1}(c)$, and
\item[b)] there exists an ordering~$(c_1, \dots, c_{2k})$ of the colors~$1, \dots, 2k$ such that every~$(v,w) \in A$ satisfies~$v \in \chi^{-1}(c_i)$ and~$w \in \chi^{-1}(c_j)$ for some~$i<j$.
\end{enumerate}
\end{definition}
Consider the following auxiliary problem.

\begin{center}
	\begin{minipage}[c]{.9\linewidth}
          \textsc{Colored $(\Pi_0+e)$-Skeleton BNSL}\\
          \textbf{Input}: A set of vertices~$N$, local scores~$\Fa=\{f_v \mid v \in N\}$, two integers~$t,k \in \mathds{N}$, and a coloring~$\chi: N \rightarrow \{1, \dots, 2k\}$.\\
          \textbf{Question}: Is there an~$(N,\Fa,t)$-valid arc set~$A \subseteq N \times N$ that is color-loyal for~$\chi$ and~$|A| \leq k$?
	\end{minipage}
\end{center}

Recall that the intuitive idea behind the color coding algorithm is to randomly color the vertices of~$N$ with~$2k$ colors and find a solution~$A$ satisfying a constraint regarding the random coloring. \textsc{Colored $(\Pi_0+e)$-Skeleton BNSL} is the problem that we solve after we randomly choose the coloring. The correspondence between \textsc{$(\Pi_0+e)$-Skeleton BNSL} and its colored version is stated in the following proposition.

\begin{proposition} \label{Prop: Ba-BNSL eq Colored Ba-BNSL}
Let~$I=(N, \Fa, t, k)$ be an instance of \textsc{$(\Pi_0+e)$-Skeleton BNSL} 
If~$I$ is a yes-instance of~\textsc{$(\Pi_0+e)$-Skeleton BNSL}, then there exist at least~$(2k)!(2k)^{(n-2k)}$ colorings~$\chi:N \rightarrow \{1,2, \dots, 2k\}$ such that~$(N, \Fa,t,k,\chi)$ is a yes-instance of~\textsc{Colored $(\Pi_0+e)$-Skeleton BNSL}.
\end{proposition}

\iflong
\begin{proof}
Let~$I$ be a yes-instance of~\textsc{$(\Pi_0+e)$-Skeleton BNSL}. Then, there exists an~$(N,\Fa,t)$-valid arc set~$A$ with~$|A| \leq k$. Observe that~$|A| \leq k$ implies that there are at most~$2k$ vertices of~$N$ that are endpoints of arcs in~$A$. 

We define a set~$\mathbb{X}$ of colorings of~$N$, such that~$\chi \in \mathbb{X}$ if~$\chi$ assigns all vertices incident with arcs of~$A$ to pairwise distinct colors. Since at most~$2k$ vertices are endpoints of arcs in~$A$, we conclude that~$|\mathbb{X}|\geq(2k)!(2k)^{(n-2k)}$. Let~$I':=(N,\Fa,t,k,\chi)$ be an instance of~\textsc{Colored $(\Pi_0+e)$-Skeleton BNSL} for some arbitrary~$\chi \in \mathbb{X}$. We show that~$A$ is a solution of~$I'$. Note that~$A$ is~$(N, \Fa, t)$-valid, so it remains to show that~$A$ is color-loyal for~$\chi$.

Since all endpoints of arcs in~$A$ have pairwise distinct colors under~$\chi$, at most one vertex in each color class~$\chi^{-1}(c)$ has a non-empty parent set. Let~$D':=(N',A)$ be the subgraph we obtain when removing all isolated vertices from~$(N,A)$. Then,~$D'$ has at most~$2k$ vertices. Consider a topological ordering~$\tau:=(v_1, \dots, v_{|N'|})$ of the DAG~$D'$. Since the vertices of~$N'$ belong to pairwise different color classes, there exists a color sequence~$(c_1, \dots, c_{2k})$ where~$v_i \in \chi^{-1}(c_i)$ for all~$i \in \{1, \dots, |N'|\}$. Let~$(v,w) \in A$. Since~$\tau$ is a topological ordering of~$D'$ we have~$v \in \chi^{-1}(c_i)$ and~$w \in \chi^{-1}(c_j)$ for some~$i<j$ Thus,~$A$ is color loyal for~$\chi$. Consequently,~$I'$ is a yes-instance of~\textsc{Colored $(\Pi_0+e)$-Skeleton BNSL}. $\hfill \Box$
\end{proof}
\fi

\iflong We next show that~\textsc{Colored $(\Pi_0+e)$-Skeleton BNSL} parameterized by~$k$ is fixed-parameter tractable.\fi

\begin{proposition} \label{Prop: Colored BA-BNSL Dyn Prog}
\textsc{Colored $(\Pi_0+e)$-Skeleton BNSL} can be solved in~$\Oh(4^k k^2 n^2 \delta_\Fa)$ time.
\end{proposition}

\begin{proof}
\iflong
Let~$I=(N,\Fa,t,k,\chi)$ be an instance of~\textsc{Colored $(\Pi_0+e)$-Skeleton BNSL} with~$f_v(\emptyset)=0$ for every~$v\in N$ and let~$C:=\{1,2, \dots, 2k\}$ denote the set of colors. By Proposition~\ref{Prop: Translation}, every instance of a constrained BNSL problem can be transformed into such an instance in~$\Oh(|\Fa|) = \Oh(n^2 \cdot \delta_{\Fa})$ time.
\fi

We fill a dynamic programming table~$T$ with entries of type~$T[C',k']$ where~$C' \subseteq C$ and~$k' \in \{0,1, \dots, k\}$. Every entry stores the maximum value of~$\sum_{v \in \chi^{-1}(C')} f_v(P^A_v)$ over all possible DAGs~$D=(N,A)$, where~$A \subseteq \chi^{-1}(C') \times \chi^{-1}(C')$ is color-loyal for~$\chi$ and contains at most~$k'$ arcs. We set~$T[\{c\},k'] := \sum_{w \in \chi^{-1}(c)} f_w(\emptyset) \iflong = 0 \fi$ for every~$c \in C$ and~$k' \in \{0,1, \dots,2k\}$. The recurrence to compute the entry for~$C' \subseteq C$ with~$|C'| > 1$ is
\begin{align*}
&T[C',k'] = \max_{c \in C'} \max_{v \in \chi^{-1}(c)} \max_{\substack{P \in \mathcal{P}_\Fa(v) \\ |P| \leq k' \\ \chi(P) \subseteq C' \setminus \{c\}}} T[C' \setminus \{c\}, k' - |P|] + f_v(P) + \sum_{w \in \chi^{-1}(c) \setminus \{v\}} f_w (\emptyset) \text{.} 
\end{align*}
The result can be computed by checking if~$T[C,k] \geq t$. Note that the corresponding network can be found via traceback. The correctness proof is straightforward and thus omitted. 
We next consider the running time. The size of~$T$ is~$\Oh(2^{2k} \cdot k)$. Note that~$\sum_{w \in \chi^{-1}(c) \setminus \{v\}} f_w (\emptyset)=0$ since we applied the preprocessing from Proposition~\ref{Prop: Translation}. Therefore, each entry can be computed in~$\Oh(2k \cdot n^2 \cdot \delta_\Fa)$ time by iterating over all~$2k$ possible colors~$c$, all~$\Oh(n)$ vertices~$v$ in the corresponding color class, and all~$\Oh(\delta_{\Fa}n)$ vertices in possible parent sets of~$v$. Altogether, \textsc{Colored $(\Pi_0+e)$-Skeleton BNSL} can be solved in~$\Oh(4^k k^2 n^2 \delta_\Fa)$~time.
 $\hfill \Box$
\end{proof}

Propositions~\ref{Prop: Ba-BNSL eq Colored Ba-BNSL} and~\ref{Prop: Colored BA-BNSL Dyn Prog} give the following.

\begin{theorem} \label{Theorem: ColorCoding Algo BA-BNSL}
There exists a randomized algorithm for~\textsc{$(\Pi_0+e)$-Skeleton BNSL} that, in $\Oh((2e)^{2k} \cdot k^2 n^2 \delta_\Fa)$~time returns \emph{no}, if given a no-instance and returns \emph{yes} with probability at least~$1-\frac{1}{e}$, if given a yes-instance.
\end{theorem}

\iflong
\begin{proof}
\textit{Algorithm.} We describe the randomized algorithm applied on an instance~$I=(N,\Fa,t,k)$. Repeat the following two steps~$e^{2k}$ times independently:
\begin{enumerate}
\item[1.] Color every vertex of~$N$ independently with one color from the set~$\{1, \dots, 2k\}$ with uniform probability. Let~$\chi: N \rightarrow \{1, \dots, 2k\}$ be the resulting coloring.
\item[2.] Apply the algorithm behind Proposition~\ref{Prop: Colored BA-BNSL Dyn Prog} to decide if~$(N,\Fa,t,k, \chi)$ is a yes-instance of \textsc{Colored $(\Pi_0+e)$-Skeleton~BNSL}. If this is the case, then return \emph{yes}.
\end{enumerate}
If for none of the~$e^{2k}$ applications the answer~\emph{yes} was returned in Step~2, then return \emph{no}.

\textit{Running Time.} We first consider the running time of the algorithm. By Proposition~\ref{Prop: Colored BA-BNSL Dyn Prog}, one application of the algorithm described above can be performed in~$\Oh(2^{2k} \cdot k^2 n^2 \delta_\Fa)$ time. Thus, the overall running time of the algorithm is~$\Oh((2e)^{2k} \cdot k^2 n^2 \delta_\Fa)$ as claimed. 

\textit{Error Probability.} We next consider the error probability of the algorithm. Given a no-instance, there exists no~$(N,\Fa,t)$-valid arc set~$A$ with~$|A| \leq k$ and therefore the answer \emph{no} is always returned in Step 2. Conversely, given a yes-instance~$I$, we conclude from Proposition~\ref{Prop: Ba-BNSL eq Colored Ba-BNSL} that there exist at least~$(2k)!(2k)^{(n-2k)}$ colorings~$\chi$ such that~$(N,\Fa,t,k,\chi)$ is a yes-instance of~\textsc{Colored $(\Pi_0+e)$-Skeleton BNSL}. The probability of randomly choosing such coloring~$\chi$ is at least
\begin{align*}
\frac{(2k)!(2k)^{(n-2k)}}{(2k)^n} \geq e^{-2k}.
\end{align*} Hence, by repeating the algorithm independently~$e^{2k}$ times, the algorithm returns \emph{yes} with a constant probability of~$1-\frac{1}{e}$. $\hfill \Box$
%
%
\end{proof}
\fi

The \iflong randomized algorithm from Theorem~\ref{Theorem: ColorCoding Algo BA-BNSL} \else algorithm \fi can be derandomized with standard techniques~\cite{NSS95,CFKLMPPS15}.

\begin{corollary} \label{Cor: BA-BNSL FPT} 
\textsc{$(\Pi_0+e)$-Skeleton BNSL} can be solved in~$(2e)^{2k} \cdot k^{\Oh(\log(k))} \cdot |I|^{\Oh(1)}$ time.
\end{corollary}

%
%
%


Bounding the number of arcs appears to be not so relevant for practical use. However, the algorithm might be useful as a heuristic upper bound: If we want to add a restricted number of dependencies to a given Bayesian network, the result of \textsc{$(\Pi_0+e)$-Skeleton BNSL} gives an upper bound for the profit we can expect from that modification. 

\subsection{\textsc{$\mathbf{(\Pi_0+e)}$-Moral BNSL}} \label{Subsection: Bounded Edges Moral}

We now study a version of~\textsc{BNSL}, where we aim to learn a network whose moralized graph has a bounded number of edges.
Formally, this is the constrained BNSL problem~\textsc{$(\Pi_0 + e)$-Moral BNSL}, where~$\Pi_0$ is the class of edgeless graphs.

Note that, given some~$k \in \mathds{N}$, a DAG~$D=(N,A)$ with~$\Mo(D) \in (\Pi_0 + ke)$ is a DAG whose moralized graph contains at most~$k$ edges. Observe that there is a simple XP-time algorithm that solves~\textsc{$(\Pi_0 + e)$-Moral BNSL} when parameterized by~$k$: Let~$I=(N,\Fa,t,k)$ be an instance of~\textsc{$(\Pi_0 + e)$-Moral BNSL}. If~$(N,A)$ is a Bayesian network whose moralized graph has~$k$ or less edges, then~$|A| \leq k$. We can find~$A$ by iterating over all~$\Oh(n^{2k})$ possible arc sets~$A$ with~$|A| \leq k$. If we consider the superstructure~$S_{\vec{\Fa}}=(N,A_{\Fa})$ we can instead iterate over all possible subsets~$A'\subseteq A_\Fa$ with~$|A'| \leq k$. Afterwards, we check if~$A'$ is~$(N,\Fa,t)$-valid and if~$\Mo(N,A') \in \Pi_0 + ke$. This implies the following.

\begin{proposition} \label{Prop: Bounded-Arc BNSL in XP}
\textsc{$(\Pi_0 + e)$-BNSL} can be solved in~$m^k \cdot |I|^{\Oh(1)}$~time.
\end{proposition}
\fi
To put this simple XP-time algorithm into context, we show that~\iflong{}\textsc{$(\Pi_0 + e)$-BNSL} is~W[1]-hard when parameterized by~$t+k$. Hence, \fi there is little hope to solve~\textsc{$(\Pi_0 + e)$-BNSL} in time~$g(t+k) \cdot |I|^{\Oh(1)}$ for any computable function~$k$.

\begin{theorem} \label{Theorem: Bounded-Edges BNSL W[1]-h}
\textsc{$(\Pi_0 + e)$-Moral BNSL} is W[1]-hard when parameterized by~$t+k$, even when~$S_{\vec{\Fa}}$ is a DAG, the maximum parent set size is 3, and every local score is either 1 or 0.
\end{theorem}
\iflong
\begin{proof}
We prove W[1]-hardness by giving a parameterized reduction from the following problem.
\begin{center}
	\begin{minipage}[c]{.9\linewidth}
          \textsc{Multicolored Clique}\\
          \textbf{Input}: A properly $\ell$-colored undirected graph $G=(V,E)$ with color classes~$C_1,\ldots,C_\ell \subseteq V$.\\
          \textbf{Question}: Is there a clique containing one vertex from each color class in $G$? 
	\end{minipage}
\end{center}
\textsc{Multicolored Clique} is  W[1]-hard when parameterized by~$\ell$~\cite{P03,FHRV09}. 

\textit{Construction.} Let~$G=(V,E)$ be a properly~$\ell$-colored undirected graph with color classes $C_1, \dots, C_\ell$. We describe how to construct an equivalent instance~$I=(N,\Fa,t,k)$ of~\textsc{$(\Pi_0 + e)$-Moral BNSL} from~$G$. First, we define the vertex set~$N$. Every vertex~$v \in V$ becomes a vertex in~$N$ and for every pair~$\{C_i, C_j\}$ ($i \neq j$) of color classes we add a vertex~$w_{\{i,j\}}$ to~$N$. Let~$W$ be the set of all such vertices~$w_{\{i,j\}}$. Moreover, we add a vertex~$x$ to~$N$ which we will call the~\emph{central vertex} for the rest of the proof. 

Second, we define the local scores~$\Fa$. For every vertex~$u \in V \cup \{x\}$ and every~$P \subseteq N \setminus \{u\}$, we set~$f_{u}(P) := 0$. It remains to define the local scores for the vertices in~$W$. Let~$i, j \in \{1, \dots, \ell\}$ with~$i \neq j$. We set~$f_{w_{\{i,j\}}}(\{u,v,x\}):=1$ if there is an edge~$\{u,v\} \in E$ connecting a vertex~$u \in C_i$ and~$v \in C_j$. For all other sets~$P$, we set~$f_{w_{\{i,j\}}}(P):=0$. Observe that the value of the local scores is either~$0$ or~$1$ and that  there are exactly~$|E|$ values with~$f_{v}(P)=1$. 

Finally, we set~$t:= \binom{\ell}{2}$ and~$k:=4 \binom{\ell}{2} + \ell$. Note that~$t+k \in \Oh(\ell^2)$. 

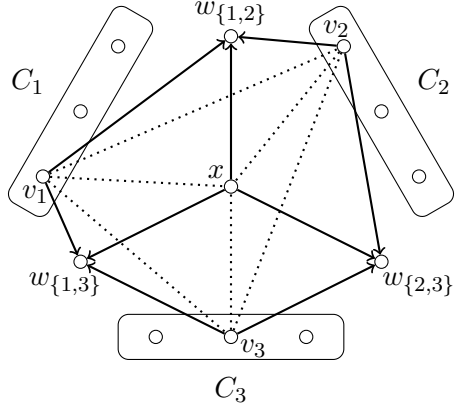
\begin{figure}
\begin{center}
\begin{tikzpicture}
\tikzstyle{knoten}=[circle,fill=white,draw=black,minimum size=5pt,inner sep=0pt]
\tikzstyle{bez}=[inner sep=0pt]

\draw[rounded corners] (-1.5, -0.3) rectangle (1.5, 0.3) {};
\node[knoten]  at (-1,0) {};
\node[knoten] (v3) at (0,0) {};
\node[knoten]  at (1,0) {};

\begin{scope}[xshift=-2cm, yshift=3cm, rotate around={60:(0,0)}]
\draw[rounded corners] (-1.5, -0.3) rectangle (1.5, 0.3) {};
\node[knoten] (v1) at (-1,0) {};
\node[knoten]  at (0,0) {};
\node[knoten]  at (1,0) {};
\end{scope}

\begin{scope}[xshift=2cm, yshift=3cm, rotate around={-60:(0,0)}]
\draw[rounded corners] (-1.5, -0.3) rectangle (1.5, 0.3) {};
\node[knoten] (v2) at (-1,0) {};
\node[knoten]  at (0,0) {};
\node[knoten]  at (1,0) {};
\end{scope}

\node[knoten] (x) at (0,2) {};
\node[bez] at (-0.2,2.2) {$x$};

\node[knoten] (w13) at (-2,1) {};
\node[bez] at (-2.2,0.7) {$w_{\{1,3\}}$};

\node[knoten] (w23) at (2,1) {};
\node[bez] at (2.5,0.7) {$w_{\{2,3\}}$};

\node[knoten] (w12) at (0,4) {};
\node[bez] at (0,4.3) {$w_{\{1,2\}}$};

\draw[->, thick]  (v1) to (w13);
\draw[->, thick]  (v3) to (w13);
\draw[->, thick]  (x) to (w13);

\draw[->, thick]  (v1) to (w12);
\draw[->, thick]  (v2) to (w12);
\draw[->, thick]  (x) to (w12);

\draw[->, thick]  (v3) to (w23);
\draw[->, thick]  (v2) to (w23);
\draw[->, thick]  (x) to (w23);


\draw[dotted, thick]  (v1) to (v2);
\draw[dotted, thick]  (v3) to (v2);
\draw[dotted, thick]  (v1) to (v3);

\draw[dotted, thick]  (v1) to (x);
\draw[dotted, thick]  (v2) to (x);
\draw[dotted, thick]  (v3) to (x);

\node[bez] at (0,-0.7) {$C_3$};
\node[bez] at (-2.7,3.4) {$C_1$};
\node[bez] at (2.7,3.4) {$C_2$};

\node[bez] at (-2.6,1.9) {$v_1$};
\node[bez] at (1.4,4.1) {$v_2$};
\node[bez] at (0.3,-0.15) {$v_3$};

\end{tikzpicture}
\end{center}
\caption{\small{An example of the construction given in the proof of Theorem \ref{Theorem: Bounded-Edges BNSL W[1]-h}. The original instance contains a multicolored clique on the vertices~$v_1 \in C_1$, $v_2 \in C_2$, and~$v_3 \in C_3$. The directed edges represent the arcs of a~DAG with score~$3$ such that the moralized graph contains~$15$ edges. The dotted edges correspond to the moralized~edges.}}\label{Figure: Pi0 W1-h}
\end{figure}
Observe that the maximum parent set size is 3 and the superstructure~$S_{\vec{\Fa}}$ is a DAG since every vertex in~$V \cup \{x\}$ has in-degree 0 in~$S_{\vec{\Fa}}$ and every vertex in~$W$ is a sink in~$S_{\vec{\Fa}}$. Figure~\ref{Figure: Pi0 W1-h} shows an example of the construction.

\textit{Intuition.} Before we show the correctness of the reduction, we start with some intuition. To obtain score~$t= \binom{\ell}{2}$, every vertex in~$W$ must choose a parent set with score~$1$. Hence, every~$w_{\{i,j\}}$ chooses a parent set~$\{u,v,x\}$ with~$u \in C_i$ and~$v \in C_j$. This choice represents the choice of an edge~$\{u,v\} \in E$ between the vertices~$u$ and~$v$ of a multicolored clique in~$G$. Considering the moralized graph of the resulting Bayesian network, the bound of the number of edges gives a bound of the number of moral edges that are incident with the central vertex~$x$. This guarantees that the chosen edges form a multicolored clique in the following sense: If the parent sets of vertices in~$W$ do not correspond to the edges of a multicolored clique in~$G$, then the moralized graph has more than~$k=4 \binom{\ell}{2} + \ell$ edges.
\iflong

\textit{Correctness.} We show that~$G$ is a yes-instance of \textsc{Multicolored Clique} if and only if~$(N,\Fa,t,k)$ is a yes-instance of~\textsc{$(\Pi_0+e)$-Moral BNSL}.

$(\Rightarrow)$ Let~$S:=\{v_1, \dots, v_\ell\}$ with~$v_i \in C_i$ be a multicolored clique in~$G$. We define the arc set~$A:=\{ (v_i, w_{\{i,j\}}), (v_j, w_{\{i,j\}}), (x, w_{\{i,j\}}) \mid w_{\{i,j\}} \in W \}$. We show that~$A$ is~$(N,\Fa,t)$-valid and that there are at most~$k$ edges in~$\Mo(N,A)$. 

Since the vertices of~$S$ are pairwise adjacent in~$G$, it holds that~$f_{w_{\{i,j\}}}(P^A_{w_{\{i,j\}}})=1$ for every~$w_{\{i,j\}} \in W$ and therefore~$\sum_{v \in N} f_v(P^A_v) = \binom{\ell}{2} = t$. Moreover, since~$S_{\vec{\Fa}}$ is a DAG, we conclude that~$(N,A)$ is a DAG. Hence,~$A$ is~$(N,\Fa,t)$-valid.

It remains to check that there are at most~$k= 4 \cdot \binom{\ell}{2} + \ell$ edges in~$\Mo(N,A)$. First, we consider the number of arcs in~$(N,A)$. Since every vertex in~$W$ has~three parents in~$(N,A)$, we conclude~$|A|=3 \cdot \binom{\ell}{2}$. Next, we consider the moral edges in~$\Mo(N,A)$. Let~$a, b \in N$ be two vertices that have a common child in~$(N,A)$. Observe that all vertices in~$V \setminus S \cup W$ have out-degree 0 in~$(N,A)$. We conclude~$a,b \in S \cup \{x\}$. Then, there are at most~$|\{\{a,b\} \mid a \in S, b \in S \cup \{x\}\}| = \binom{\ell}{2} + \ell$ moral edges. Hence, there are at most~$k= 4 \cdot \binom{\ell}{2} + \ell$ edges in~$\Mo(N,A)$.

$(\Leftarrow)$ Conversely, let~$A\subseteq N \times N$ be an~$(N,\Fa,t)$-valid arc set such that~$\Mo(N,A)$ contains at most~$k= 4 \binom{\ell}{2} + \ell$ edges. We show that there exists a multicolored clique~$S$ in~$G$.

Since~$A$ is~$(N,\Fa,t)$-valid, we know that~$\sum_{v \in N} f_v(P^A_v) = \binom{\ell}{2}$ and thus~$f_{w_{\{i,j\}}}(P^A_{w_{\{i,j\}}})=1$ for every~$w_{\{i,j\}} \in W$. By the construction of~$\Fa$ this implies~$|P^A_{w_{\{i,j\}}}|=3$ for every~$w_{\{i,j\}} \in W$. We conclude~$|A|= 3 \binom{\ell}{2}$. Hence, there are at most~$\binom{\ell}{2}+\ell$ moral edges in~$\Mo(N,A)$. 

Before we define the multicolored clique~$S$, we take a closer look at the moral edges that are incident with vertices of the color classes~$C_1, \dots, C_\ell$. Let~$C_i$ and $C_j$ be distinct color classes. Then, since~$P^A_{w_{\{i,j\}}}$ contains one vertex from~$C_i$ and one vertex from~$C_j$, there exists a moral edge between the vertices of~$C_i$ and~$C_j$. Hence, there are at least~$\binom{\ell}{2}$ moral edges between the color classes of~$C_1, \dots, C_\ell$. Now, since the overall number of moral edges in~$\Mo(N,A)$ is at most~$\binom{\ell}{2}+\ell$, we may conclude that there are at most~$\ell$ moral edges that are incident with the central vertex~$x$. We use the following claim to define a multicolored clique~$S$ in~$G$.

\begin{claim} \label{Claim: One Moral edge per color class}
For every color class~$C_i$ it holds that~$|E_{\Mo(N,A)}(C_i,\{x\})| = 1$.
\end{claim}

\begin{proof}
Let~$C_i$ be a color class. Note that there is no arc in~$A$ connecting~$x$ with some vertices in~$C_i$. So, $E_{\Mo(N,A)}(C_i,\{x\})$ contains only moral edges. For every~$j \in \{1, \dots, \ell\}$ with~$j \neq i$, the vertex~$w_{\{i,j\}}$ has a parent set~$P^A_{w_{\{i,j\}}}$ containing some~$v \in C_i$, $u \in C_j$ and~$x$. 
Then, there exist moral edges~$\{u,v\}$, $\{v,x\}$, and~$\{u,x\}$. Therefore, every color class contains a vertex that is adjacent to~$x$ by a moral edge. Since there are at most~$\ell$ moral edges incident with~$x$, we conclude~$|E_{\Mo(N,A)}(C_i,\{x\})| = 1$. $\hfill \Diamond$
\end{proof}

We now define~$S:=\{v_1, v_2, \dots, v_\ell\}$, where~$v_i$ is the unique element in~$E_{\Mo(N,A)}(C_i,\{x\})$. Observe that this implies~$P^A_{w_{\{i,j\}}}=\{v_i, v_j, x\}$ for all~$v_i,v_j \in S$ with~$i \neq j$. We show that~$S$ is a multicolored clique in~$G$. Obviously, the vertices of~$S$ are elements of distinct color classes. Thus, it remains to show that the vertices in~$S$ are pairwise adjacent in~$G$. Let~$v_i, v_j \in S$ with~$i \neq j$. Then,~$P^A_{w_{\{i,j\}}}=\{v_i, v_j, x\}$ and since~$f_{w_{\{i,j\}}}(P^A_{w_{\{i,j\}}})=1$ it follows from the construction of~$\Fa$ that there is an edge~$\{v_i, v_j\} \in E$. Hence, $S$ is a multicolored clique in~$G$.
\else

Due to lack of space, the correctness proof is deferred to a full version of this work. 
\fi
$\hfill \Box$
\end{proof}
\fi

\section{BNSL with Bounded Feedback Edge Set}

In this section, we provide a first step into the study of the parameterized complexity of learning a Bayesian network whose moralized graph has a feedback edge set of bounded size. Formally, this is the constrained BNSL problem~\textsc{$(\Pi_F+e)$-Moral BNSL}, where~$\Pi_F$ is the class of forest, which are undirected acyclic graphs. 
Recall that for efficient inference it is desirable to have a small treewidth in the moralized graph~\cite{D09}. As all other parameters considered in this work, the size of a feedback edge set is a upper bound for the treewidth. Thus, learning a network where the moralized graph has a bounded feedback edge set is motivated from a practical point of~view. 

Before we consider~\textsc{$(\Pi_F+e)$-Moral BNSL}, we briefly discuss~\textsc{$(\Pi_F+e)$-Skeleton BNSL}. When~$k=0$, this is the problem of learning a Bayesian network with an acyclic skeleton, also known as polytree. Finding an optimal polytree is NP-hard even on instances with maximum parent set size 2~\cite{D99}. Consequently, \textsc{$(\Pi_F+e)$-Skeleton BNSL} is NP-hard even if~$k=0$ is fixed. In contrast, the case~$k=0$ can be solved efficiently if we consider the moralized graph instead of the skeleton. This can be seen as follows. 
Let~$D:=(N,A)$ be a DAG such that~$\Mo(D)$ is acyclic. Then, each~$v \in N$ has at most one parent in~$D$, since otherwise~$\Mo(D)$ contains a triangle. Thus,~$D:=(N,A)$ is a branching.
Consequently, \textsc{$(\Pi_F+e)$-Moral BNSL} with~$k=0$ can be solved by computing an optimal branching which can be done in polynomial time~\cite{CL68,GKLOS15}.

\begin{proposition}
\textsc{$(\Pi_F+e)$-Moral BNSL} can be solved in polynomial time when limited to instances with~${k=0}$.
\end{proposition}


This positive result makes it interesting to study the parameterized complexity of~\textsc{$(\Pi_F+e)$-Moral BNSL} when parameterized by~$k$. In the following, we provide a first step into this parameterized complexity analysis and show that~\textsc{$(\Pi_F+e)$-Moral BNSL} is W[1]-hard when parameterized by~$k$. Thus,~\textsc{$(\Pi_F+e)$-Moral BNSL} can presumably not be solved in~$g(k) \cdot |I|^{\Oh(1)}$~time for a computable function~$g$. However, an XP-time algorithm might still be possible.

\begin{theorem} \label{Theorem: B-FES BNSL W1h}
\textsc{$(\Pi_F+e)$-Moral BNSL} is W[1]-hard when parameterized by~$k$, even when \iflong restricted to instances where\fi~$S_{\vec{\Fa}}$ is a DAG and the maximum parent set size is 4.
\end{theorem}

\iflong
\begin{proof}
We give a parameterized reduction from~\textsc{$(\Pi_0+e)$-Moral BNSL} parameterized by the number of edges~$k$ which is~W[1]-hard even on instances where the superstructure is a DAG and the maximum parent set size is 3 due to Theorem~\ref{Theorem: Bounded-Edges BNSL W[1]-h}.

\textit{Construction.} Let~$I:=(N,\Fa,t,k)$ be such an instance of~\textsc{$(\Pi_0+e)$-Moral BNSL}. We describe how to construct an equivalent instance~$I':=(N',\Fa',t',k')$ of~\textsc{$(\Pi_F+e)$-Moral BNSL} where~$k'=k$. The vertex set is~$N':=N \cup \{x\}$ for some~$x\not \in N$. 

We define the vertex set by~$N':=N \cup \{x\}$ for some~$x\not \in N$. 
To define the local scores~$\Fa'$, we set~$\ell^+ := 1+ \sum_{v \in N} \max_{P \subseteq N \setminus \{v\}} f_v(P)$. For every~$v \in N$ we set~$f_v'(P)= f_v(P \setminus \{x\}) + \ell^+$ if~$x \in P$ and~$P \setminus \{x\} \in \mathcal{P}_\Fa(v)$. In all other cases, we set~$f_v'(P)=0$. For the vertex~$x$, we set~$f_x'(P)=0$ for every~$P$. Finally, we set~$t':= t + n \cdot \ell^+$.

We can obviously compute~$I'$ from~$I$ in polynomial time. Since~$I$ is an instance where~$S_{\vec{\Fa}}$ is a DAG and the maximum parent set size is 3, we conclude that the maximum parent set size of~$I'$ is 4 and that~$S_{\vec{\Fa'}}$ is a DAG.

\textit{Intuition.} Before we prove the correctness of the reduction we provide some intuition. To obtain an~$(N',\Fa',t')$-valid arc set~$A'$, the vertex~$x$ must be a parent of every vertex of~$N$. Hence, for every~$v \in N$, there exists an edge~$\{x,v\}$ in~$\Mo(N',A')$. The idea is that~$\Mo(N',A')$ can be transformed into an acyclic graph by deleting all edges between the vertices of~$N$.

\textit{Correctness.} We now prove that~$I$ is a yes-instance of~\textsc{$(\Pi_0+e)$-Moral BNSL} if and only if~$I'$ is a yes-instance of~\textsc{$(\Pi_F+e)$-Moral BNSL}.

$(\Rightarrow)$ Let~$A \subseteq N \times N$ be an~$(N,\Fa,t)$-valid arc set such that~$\Mo(D)$ for~$D:=(N,A)$ contains at most~$k$ edges. We then define~$A':= A \cup \{(x,v) \mid v \in N\}$ and let~$D':=(N',A')$. We show that~$A'$~is~$(N',\Fa',t')$-valid and~$\Mo(D')$ has a feedback edge set of size at most~$k$.

We first show that~$A'$ is~$(N',\Fa',t')$-valid. Since~$S_{\vec{\Fa'}}$ is a DAG we conclude that~$D$ is a~DAG. Moreover,~$P^{A'}_v=P^A_v \cup \{x\}$ for every~$v \in N$ and therefore
\begin{align*}
\sum_{v \in N'} f_v'(P^{A'}_v) &= \sum_{v \in N} (f_v(P^{A}_v) + \ell^+)\\
&=t+n \cdot \ell^+ = t'.
\end{align*}
Consequently, $D'$ is~$(N',\Fa',t')$-valid. It remains to show that~$\Mo(D')$ has a feedback edge set of size at most~$k$. To this end, consider the following claim.

\begin{claim} \label{Claim: Moral Edges iff}
Let~$v,w \in N$. Then,~$\{v,w\}$ is a moral edge in~$\Mo(D)$ if and only if~$\{v,w\}$ is a moral edge in~$\Mo(D')$.
\end{claim}
\begin{proof}
Let~$\{v,w\}$ be a moral edge in~$\Mo(D)$. Then, there exists a vertex~$u \in N$ such that~$(v,u)$, $(w,u) \in A$. Since~$A \subseteq A'$ we conclude that~$\{v,w\}$ is a moral edge in~$\Mo(D')$.

Conversely, let~$\{v,w\}$ be a moral edge in~$\Mo(D')$. Then,~$v$ and~$w$ have a common child~$u$ in~$D'$. Since~$x$ has no incoming arcs, we conclude~$u \in N$ and therefore~$(v,u),(w,u) \in A$. Hence,~$\{v,w\}$ is a moral edge in~$\Mo(D)$. $\hfill \Diamond$
\end{proof}

Claim~\ref{Claim: Moral Edges iff} together with the fact that~$(N \times N) \cap A' = A$ implies that~$v,w \in N$ are adjacent in~$\Mo(D)$ if and only if they are adjacent in~$\Mo(D')$. Hence, if we delete every edge of~$\Mo(D)$ from~$\Mo(D')$ we obtain the graph~$G:=(N', \{\{x,v\} \mid v \in N\})$ which is acyclic. Since there are at most~$k$ edges in~$\Mo(D)$ we conclude that there exists a feedback edge set of size at most~$k$ for~$\Mo(D')$.

$(\Leftarrow)$ Conversely, let~$A'$ be an~$(N',\Fa',t')$-valid arc set such that~$\Mo(D')$ for~$D':=(N',A')$ has a feedback edge set of size at most~$k$. We define~$A:= (N \times N) \cap A'$. Note that~$P_v^A = P_v^{A'} \setminus \{x\}$ for every~$v \in N$.

We first show that~$D:=(N,A)$ is~$(N,\Fa,t)$-valid. Obviously,~$D$ is a DAG since~$S_{\vec{\Fa}}$ is a DAG. Moreover, it holds that
\begin{align*}
\sum_{v \in N} f_v(P^{A}_v) &= \sum_{v \in N} f_v(P^{A'}_v \setminus \{x\})\\
&=t'-n \cdot \ell^+ = t.
\end{align*}
Consequently, $D$ is~$(N,\Fa,t)$-valid. It remains to show that there are at most~$k$ edges in~$\Mo(D)$. To this end, observe that~$x \in P^{A'}_v$ for every~$v \in N$: If there exists a vertex~$w \in N$ with~$x \not \in P^{A'}_w$, then~$f'_w(P^{A'}_w) = 0$ and therefore the sum of the local scores is smaller than~$n \cdot \ell^+$. This contradicts the fact that~$A'$ is~$(N',\Fa',t')$-valid. 

%

Next, assume towards a contradiction that there are more than~$k$ edges in~$\Mo(D)$. Since~$A \subseteq A'$, this implies that in~$\Mo(D')$ there are more than~$k$ edges between the vertices of~$N$. Furthermore, since~$x \in P^{A'}_v$ for every~$v \in N$ we conclude that every vertex in~$N$ is adjacent to~$x$ in~$\Mo(D')$. Hence, $\Mo(D')$ consists of~$n+1$ vertices and at least~$n+k+1$ edges which contradicts the fact that~$\Mo(D')$ has a feedback edge set of size at most~$k$. $\hfill \Box$
\end{proof}
\fi

\section{On Problem Kernelization}
In this section we prove a new hardness result for \textsc{Vanilla-BNSL}: We show that under the standard assumtion~$\text{NP} \not \subseteq \text{coNP} / \poly$, \textsc{Vanilla-BNSL} does not admit a polynomial problem kernel when parameterized by the number of vertices. That is, it is presumably impossible to transform an instance of \textsc{Vanilla-BNSL} in polynomial time into an equivalent instance of size~$|I| = n^{\Oh(1)}$. Thus, it is sometimes necessary to keep an exponential number of parent scores to compute an optimal network. The kernel lower bound even holds for instances where all local scores are either 0 or 1. Thus, the kernel lower bound is not based on the fact that large local scores might be incompressible.

We then use the kernel lower bound for~\textsc{Vanilla-BNSL} to complement the FPT result from Corollary \ref{Cor: BA-BNSL FPT} and show that there is little hope that \textsc{$(\Pi_0+e)$-Skeleton BNSL} admits a polynomial problem kernel.

\begin{theorem} \label{Theorem: No Poly Kernel BA-BNSL}
\textsc{Vanilla-BNSL} parameterized by~$n+t$ does not admit a polynomial kernel unless~$\badstuffhappens$ even when  restricted to instances where all local scores are either 0 or 1.
\end{theorem}

\iflong
\begin{proof}
We prove the theorem by giving a polynomial parameter transformation from the following problem.
\begin{center}
	\begin{minipage}[c]{.9\linewidth}
          \textsc{Multicolored Independent Set}\\
          \textbf{Input}: A properly $\ell$-colored undirected graph $G=(V,E)$ with color classes~$C_1,\ldots,C_\ell \subseteq V$.\\
          \textbf{Question}: Is there an independent set containing one vertex from each color in $G$? 
	\end{minipage}
\end{center}
\textsc{Multicolored Independent Set} does not admit a polynomial kernel when parameterized by~$|C_1 \cup \dots \cup C_{\ell-1}|$ unless~$\badstuffhappens$~\cite{GK20}.

\textit{Construction.} Let~$G=(V,E)$ be an instance of~\textsc{Multicolored Independent Set} with the color classes~$C_1, \dots, C_\ell$. We describe how to construct an equivalent instance~$(N,\Fa,t)$ of~\textsc{Vanilla-BNSL}. First, set~$N:= C_1 \cup C_2 \cup \dots \cup C_{\ell-1} \cup \{x\}$ for some~$x \not \in V$. Second, we define the local scores~$\Fa$ as follows: Let~$i \in \{1, \dots, \ell-1\}$. For~$v \in C_i$ we set~$f_v(P)=1$ if~$P=(C_i \setminus \{v\}) \cup (N_G(v) \setminus C_\ell) \cup \{x\}$. Otherwise, we set~$f_v(P)=0$. For~$x$, we set~$f_x(P)=1$ if there exists some~$w \in C_\ell$ with~$N_G(w)=P$. Otherwise we set~$f_x(P)=0$. Finally, we set~$t:= \ell$.

Observe that the value of the local scores is either~$1$ or~$0$, and that there are exactly~$|V|$ values where~$f_v(P)=1$. Hence,~$|\Fa| \in \Oh(|V|)$. We can obviously compute~$(N,\Fa,t)$ in polynomial time from~$G$. Furthermore, recall that~$|N|=|C_1 \cup \dots \cup C_{\ell-1}| +1$ and therefore,~$n+t$ is polynomially bounded in~$|C_1 \cup \dots \cup C_\ell|$.

\textit{Intuition:} Before we show the correctness of the polynomial parameter transformation, we start with some intuition. To reach the score~$t=\ell$, exactly one vertex per color class~$C_1, \dots, C_{\ell-1}$ and the vertex~$x$ must learn a parent set with score~$1$. The vertices from~$C_1, \dots, C_{\ell-1}$ and the choice of the parent set of~$x$ then correspond to a multicolored set in~$G$. The condition that the resulting directed graph must be a DAG guarantees that the chosen vertices form an independent~set.

\textit{Correctness.} $(\Rightarrow)$ Let~$S=\{v_1, \dots, v_\ell\}$ be a multicolored independent set in~$G$ with~$v_i \in C_i$ for all~$i \in \{1, \dots, \ell\}$. We define the arc set~$A$ by defining the parent sets of all vertices in~$N$: For all~$v \in N \setminus \{v_1, \dots, v_{\ell-1}\}$ we set~$P^A_v := \emptyset$. Next, for~$v_i \in \{v_1, \dots, v_{\ell-1}\}$ we set~$P_{v_i}^A := (C_i \setminus \{v_i\}) \cup (N_G(v_i) \setminus C_\ell ) \cup \{x\}$. Finally, we set~$P^A_x=N_G(v_\ell)$. We now prove that~$A$ is~$(N,\Fa,t)$-valid. By definition of~$\Fa$ it holds that~$f_v(P^A_v)=1$ for every~$v \in \{v_1, \dots, v_{\ell-1},x\}$. Hence,~$\sum_{c \in N}f_v(P^A_v) = t$.

It remains to show that~$D:=(N,A)$ is a DAG. If~$D$ contains a directed cycle, all vertices on the directed cycle have incoming and outgoing arcs. Observe that~$v_1, \dots, v_{\ell-1}$, and~$x$ are the only vertices with incoming arcs.

We first prove that every~$v \in \{v_1, \dots, v_{\ell-1}\}$ is a sink in~$D$. Assume towards a contradiction that there is some~$v \in \{v_1, \dots, v_{\ell-1}\}$ that has an outgoing arc~$(v,w) \in A$. Without loss of generality, let~$v=v_1$. Since~$v_1 \not \in P^A_{v_1}$ and only the vertices~$v_1, v_2, \dots, v_{\ell-1}, x$ have parents under~$A$, we conclude~$w \in \{v_2, \dots, v_{\ell-1},x\}$. If~$w \in \{v_2, \dots, v_{\ell-1}\}$, then~$v \in P^A_w$ and therefore~$v \in N_G(w)$. Otherwise, if~$w=x$, then~$v \in P^A_x$ and therefore~$v \in N_G(v_\ell)$. Both cases contradict the fact that~$S$ is an independent set in~$G$ and therefore every~$v \in \{v_1, \dots,  v_{\ell-1}\}$ is a sink in~$D$.

We conclude that~$x$ is the only vertex that might have incoming and outgoing arcs in~$D$. Hence,~$x$ is the only vertex that might be part of a directed cycle. Since~$x \not \in P^A_x$ we conclude that there is no cycle in~$D$ and therefore~$D$~is~a~DAG.

$(\Leftarrow)$ Let~$A$ be an~$(N,\Fa,t)$-valid arc set. We show that there exists a multicolored independent set~$S$ in~$G$. To this end, consider the following claim.
\begin{claim} \ \label{Claim: Claim for BA-BNSL NoPolyKernel}
\begin{enumerate}
\item[a)] There are at least~$\ell$ vertices~$v \in N$ with~$f_v(P^A_v)=1$.
\item[b)] For every~$C_i$ with~$i\in \{1, \dots, \ell-1\}$ there is at most one vertex~$v_i \in C_i$ with~$f_{v_i}(P_{v_i}^A)=1$.
\end{enumerate}
\end{claim}

\begin{proof}
We first show statement~$a)$. Since~$A$ is~$(N,\Fa,t)$-valid we know that~$\sum_{v \in V} f_v(P^A_v) \geq t = \ell$. Since every local score is either~$0$ or~$1$, statement~$a)$ follows.

We next show statement~$b)$. Note that~$(N,A)$ is a DAG since~$A$ is~$(N,\Fa,t)$-valid. Let~$i \in \{1, \dots, \ell-1\}$. Assume towards a contradiction that there are distinct~$u, v \in C_i$ with~$f_{u}(P_{u}^A)=f_{v}(P_{v}^A)=1$. Then, by the construction of~$\Fa$ we conclude~$C_i \setminus \{u\} \subseteq P_{u}^A$ and $C_i \setminus \{v\} \subseteq P_{v}^A$. Hence~$(u, v), (v, u) \in A$ which contradicts the fact that~$(N,A)$ is a DAG. Consequently, statement~$b)$ holds. $\hfill \Diamond$
\end{proof}

From Claim~\ref{Claim: Claim for BA-BNSL NoPolyKernel}~$a)$ and~$b)$ we conclude that for each~$C_i$ with~$i \in \{1, \dots, \ell-1\}$ there is exactly one~$v_i \in C_i$ with~$f_{v_i}(P^A_{v_i})=1$ and that~$f_{x}(P^A_{x})=1$. Moreover,~$f_{x}(P^A_{x})=1$ implies that there exists a vertex~$v_\ell \in C_\ell$ with~$N_G(v_\ell)=P^A_x$. We define~$S:= \{v_1, \dots, v_{\ell-1}, v_\ell\}$ and show that~$S$ is a multicolored independent set in~$G$.

Obviously, the vertices of~$S$ are from pairwise distinct color classes. Thus, it remains to show that no two vertices in~$S$ are adjacent in~$G$. Assume towards a contradiction that there exist~$v,w \in S$ such that~$\{v,w\} \in E$. Without loss of generality, let~$v=v_1$. Consider the following cases.

\textbf{Case 1:} $w \in \{v_2, \dots, v_{\ell-1}\}$\textbf{.} Then, $\{v,w\} \in E$ implies~$w \in N_G(v) \setminus C_\ell$ and~$v \in N_G(w) \setminus C_\ell$. Together with the fact that~$f_v(P^A_v)=f_w(P^A_w)=1$ we conclude~$(v,w),(w,v) \in A$ which contradicts the fact that~$(N,A)$ is a~DAG.

\textbf{Case 2:} $w=v_\ell$\textbf{.} Then,~$\{v,w\} \in E$ implies~$v \in N_G(w)$ and therefore~$v \in P^A_x$. Moreover~$f_v(P^A_v)=1$ implies~$x \in P_v^A$. Hence,~$(v,w),(w,v) \in A$ contradicting the fact that~$(N,A)$ is a DAG.

We conclude that no two vertices of~$S$ are adjacent in~$G$ and therefore,~$S$ is a multicolored independent set in~$G$. $\hfill \Box$.
\end{proof}
\fi

We next use Theorem~\ref{Theorem: No Poly Kernel BA-BNSL} to complement the FPT result for~\textsc{$(\Pi_0+e)$-Skeleton BNSL} by a kernel lower-bound for constrained BNSL problems. Consider an arbitrary constrained BNSL problem for some monotone and infinite graph class~$\Pi$. Recall that, if~$k=n^2$ the sparsity constraint always holds. Thus, the constrained BNSL problem is the same as~\textsc{Vanilla-BNSL} on instances with~$k=n^2$. Together with the kernel lower-bound from Theorem~\ref{Theorem: No Poly Kernel BA-BNSL}, this implies the following.

\begin{corollary} \label{Cor: BNSL No Poly Kernel for n}
Let~$\Pi$ be a monotone graph class that contains infinitely many graphs. Then, every constrained BNSL problem for~$\Pi$ parameterized by~$k+t$ does not admit a polynomial kernel unless~$\badstuffhappens$.
\end{corollary}

\section{Conclusion}
 
 We have outlined the tractability borderline of \textsc{Bayesian Network Structure
   Learning} with respect to several structural constraints on the learned network or on
 its moralized graph. In particular, we have shown that putting structural sparsity
 constraints on the moralized graph may make the problem harder than putting similar
 constraints on the network. This is somewhat counterintuitive since the moralized graph
 is a supergraph of the underlying undirected graph of the network. It seems interesting
 to investigate this issue further, that is, to find other structural constraints such
 that putting these constraints on the network leads to an easier problem than putting
 them on the moralized graph.

 The two most important concrete questions left open by our work are the following. First,
 can we compute an optimal network where the skeleton has dissociation number at most~$k$ in polynomial
 time for constant~$k$? Second, can we compute an optimal network whose moralized
 graph has a feedback edge set of size at most~$k$ in polynomial time for constant~$k$?
 
 Another important direction for future work is to study how well the algorithms for~\textsc{$(\Pi_0+v)$-Skeleton-BNSL} (Theorem~\ref{Theorem: VC XP Algo}) and~\textsc{$(\Pi_1+v)$-Moral-BNSL} (Theorem~\ref{Theorem: XP Algo}) perform on benchmark data sets. This way, one can extend the work of~\citeA{KP15} who experimentally evaluated the performance of a similar algorithm for~\textsc{$(\Pi_0+v)$-Moral BNSL}.

 As a final remark, the algorithm for learning a Bayesian network with a bounded number of arcs (Theorem~\ref{Theorem: ColorCoding Algo BA-BNSL}) seems to have no direct practical
 applications. It may, however, be useful as subroutines in a branch-and-bound scenario
 where one may aim to compute upper or lower bounds on the score that can be achieved by adding~$k$ arcs to a network that is currently considered in the search.
 Thus, it would be interesting to explore variants of \textsc{Bayesian Network Structure
   Learning} where the input contains a partial network and the aim is to extend it. Do
 the positive results for \textsc{Bayesian Network Structure Learning} also hold for this
 more general problem?

\bibliography{bnsl}
\bibliographystyle{theapa}

\newpage

\end{document}